\RequirePackage{fix-cm} 
\documentclass[thm-restate,american,USenglish]{lipics-v2021} 
\nolinenumbers
\hideLIPIcs

\usepackage{xspace}
\usepackage{amsmath,mathtools}
\usepackage{graphicx}
\usepackage{tikz}
\usetikzlibrary{
    arrows.meta,
    positioning,
    shapes.geometric,
    calc,
    decorations.pathreplacing,
    fit,
    quotes,
}
\usepackage{ifthen}

\usepackage{booktabs,array,makecell}

\usepackage[textsize=footnotesize,backgroundcolor=orange!5,bordercolor=red!80!black]{todonotes}

\usepackage[T1]{fontenc}
\usepackage{lmodern}
\usepackage{wref}
\usepackage{slantsc}
\DeclareFontShape{T1}{lmr}{b}{sc}{<->ssub*cmr/bx/sc}{}
\DeclareFontShape{T1}{lmr}{bx}{sc}{<->ssub*cmr/bx/sc}{}

\DeclarePairedDelimiter{\lrParens}{(}{)} 
\newcommand{\lrSize}[1]{\left| #1 \right|}

\newcommand{\Oh}[2][*]{%
    \mathcal{O}\ifthenelse{\equal{#1}{*}}{\lrParens*{#2}}{\lrParens[#1]{#2}}%
}
\newcommand{\OhTheta}[2][*]{%
    \Theta\ifthenelse{\equal{#1}{*}}{\lrParens*{#2}}{\lrParens[#1]{#2}}%
}
\newcommand{\OhMega}[2][*]{%
    \Omega\ifthenelse{\equal{#1}{*}}{\lrParens*{#2}}{\lrParens[#1]{#2}}%
}
\newcommand{\logB}{\log_{B}}
\newcommand{\IO}[0]{I/O\xspace}
\newcommand{\IOs}[0]{I/Os\xspace}

\newcommand{\Operation}[1]{\textrm{\textsc{#1}}\xspace}

\newcommand{\Construct}{\Operation{Construct}}
\newcommand{\Insert}{\Operation{Insert}}
\newcommand{\Delete}{\Operation{Delete}}
\newcommand{\ChangeKey}{\Operation{ChangeKey}}
\newcommand{\QueryElement}{\Operation{QueryElement}}
\newcommand{\QueryRank}{\Operation{QueryRank}}

\newcommand{\DecreaseKey}{\Operation{DecreaseKey}}
\newcommand{\Minimum}{\Operation{Minimum}}

\newcommand{\SearchGap}{\Operation{GapByElement}}
\newcommand{\RankGap}{\Operation{GapByRank}}
\newcommand{\IncrementGap}{\Operation{GapIncrement}}
\newcommand{\DecrementGap}{\Operation{GapDecrement}}
\newcommand{\SplitGap}{\Operation{GapSplit}}

\newcommand{\MergeIntervals}{\Operation{IntervalsMerge}}
\newcommand{\InsertInterval}{\Operation{IntervalsInsert}}
\newcommand{\DeleteInterval}{\Operation{IntervalsDelete}}
\newcommand{\ChangeInterval}{\Operation{IntervalsChange}}
\newcommand{\SplitInterval}{\Operation{IntervalSplit}}

\newcommand{\Int}[1]{I_{#1}}
\newcommand{\out}[2][*]{%
    \mathbf{o}%
    \ifthenelse{\equal{#1}{*}}{\lrParens*{\Int{#2}}}{\lrParens[#1]{\Int{#2}}}%
}

\let\oldsubparagraph\subparagraph
\renewcommand{\subparagraph}[1]{\oldsubparagraph{\boldmath #1}}

\title{Towards Lazy B-Trees%
\thanks{\textbf{Erratum:} After the publication of the MFCS proceedings, we noticed that the costs of keeping external pointers to elements in the data structure are not fully taken into account in our analysis (see \wref{sec:contribution} for details).  Therefore, the use of Lazy B-Trees as external-memory priority queue with decrease-key is not currently known to be as favorable as claimed in the conference version. Specifically, for the worst case of having to maintain pointers to all elements, the I/O cost of delete-min in a straight-forward implementation increases to $O(\log N)$ \IOs. }}

\author{Casper Moldrup Rysgaard}{Aarhus University, Denmark}{rysgaard@cs.au.dk}{https://orcid.org/0000-0002-3989-123X}{Independent Research Fund Denmark, grant~9131-00113B.}

\author{Sebastian Wild}{Philipps-University Marburg, Germany, and University of Liverpool, UK}{wild@informatik.uni-marburg.de}{https://orcid.org/0000-0002-6061-9177}{EPSRC grant EP/X039447/1.}

\authorrunning{C. M. Rysgaard, S. Wild.}
\titlerunning{Towards Lazy B-Trees}

\Copyright{Casper Moldrup Rysgaard, Sebastian Wild}



\EventEditors{}
\EventNoEds{1}
\EventLongTitle{}
\EventShortTitle{}
\EventAcronym{}
\EventYear{}
\EventDate{}
\EventLocation{}
\EventLogo{}
\SeriesVolume{}
\ArticleNo{}

\begin{document}

\maketitle

\begin{abstract}
Lazy search trees (Sandlund \& Wild FOCS 2020, Sandlund \& Zhang SODA 2022) are sorted dictionaries whose update and query performance smoothly interpolates between that of efficient priority queues and binary search trees~-- automatically, depending on actual use; no adjustments are necessary to the data structure to realize the cost savings.
In this paper, we design \emph{lazy B-trees}, a variant of lazy search trees suitable for external memory that generalizes the speedup of B-trees over binary search trees wrt.\ input/output operations to the same smooth interpolation regime.

A key technical difficulty to overcome is the lack of a (fully satisfactory) external variant of \emph{biased} search trees, on which lazy search trees crucially rely.
We give a construction for a subset of performance guarantees sufficient to realize external-memory lazy search trees, which we deem of independent interest.

\end{abstract}

\section{Introduction}
\label{sec:introduction}

We introduce the \emph{lazy B-tree} data structure, which brings the adaptive performance guarantees of lazy search trees to external memory.

The binary search tree (BST) is a fundamental data structure, taught in every computer science degree and widespread in practical use.
Wherever rank-based operations are needed, e.g., finding a $k$th smallest element in a dynamic set or determining the rank of an element in the set, i.e., its position in the sorted order of the current elements, augmented BSTs are the folklore solution: binary search trees using one of the known schemes to ``balance'' them, i.e., guarantee $\Oh{\log N}$ height for a set of size $N$, where we additionally store the subtree size in each node.
On, say, AVL-trees, all operations of a sorted dictionary, i.e., rank, select, membership,
predecessor, successor, minimum, and maximum, as well as insert, delete, change-key, split, and merge
can all be supported in $\Oh{\log N}$ worst case time, where $N$ is the current size of the set.

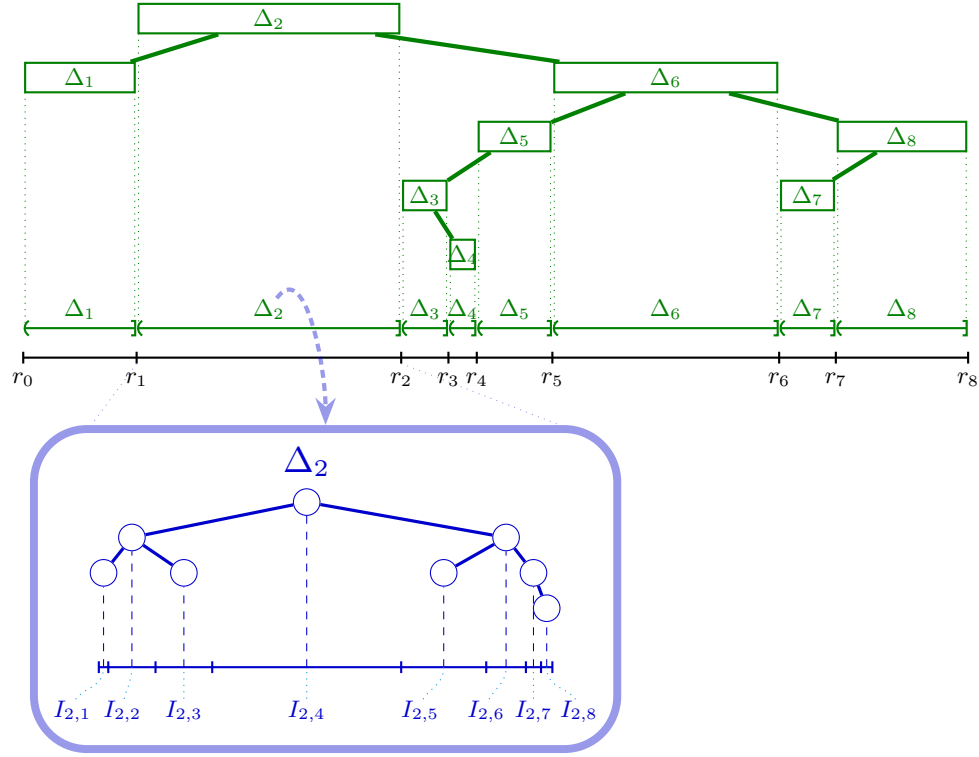
\begin{figure}[t]
	\centering
    
	\begin{tikzpicture}[
			xscale=.5,yscale=.78,
			gap/.style={green!50!black},
			interval/.style={blue!80!black},
			gap node/.style={gap,rectangle,draw,thick,inner sep=0pt,minimum width=0pt},
			interval node/.style={interval,circle,draw,minimum size=10pt},
		]
		\small
		\def\yy{-6}
		
		\foreach [count=\i] \f/\t/\d in {0/3/2,3/10/1,10/11.25/4,11.25/12/5,12/14/3,14/20/2,20/21.5/4,21.5/25/3} {
			\coordinate (dl\i) at (\f+0.05,-\d);
			\coordinate (dr\i) at (\t-0.05,-\d+0.5);
			\node[gap node,fit=(dl\i) (dr\i)] (Delta\i) {} ;
			\node[gap] at (Delta\i) {\small$\Delta_{\i}$};

			\draw[dotted,gap] (\f+0.05,-\d) -- (\f+0.05,\yy) ;
			\draw[dotted,gap] (\t-0.05,-\d) -- (\t-0.05,\yy) ;

			\draw[gap,{Parenthesis[]-Bracket[]},thick] 
				(\f+0.01,\yy) -- node[above=3pt,inner sep=0pt] (D\i) {$\Delta_{\i}$} (\t-0.01,\yy) ;
		}
		\foreach \p/\c in {2/1,2/6,6/5,5/3,3/4,6/8,8/7} {
			\draw[gap,ultra thick] (Delta\p) -- (Delta\c) ;
		}
		
		\begin{scope}[shift={(0,\yy-0.5)}]
			\draw[thick] (0,0) -- (25,0);
			\foreach [count=\i from 0] \x in {0,3,10,11.25,12,14,20,21.5,25} {
				\draw[thick] (\x,.1) -- ++(0,-0.2) node[below] (r\i) {\small $r_{\i}$};
			}
		\end{scope}
		
		\begin{scope}[shift={(2,-8.25)}]
			\node[interval,scale=1.5] at (5.5,0) {$\Delta_2$} ;
			\def\nodes{%
					0/0.25/3,%
					0.25/1.5/2,%
					1.5/3/3,%
					3/8/1,%
					8/10.25/3,%
					10.25/11.3/2,%
					11.3/11.7/3,%
					11.7/12/4%
				}
			\foreach [count=\j] \f/\t/\d in \nodes {
					\node[interval node] (I\j) at (0.5*\f+0.5*\t,-\d*.6-.1) {};
				}
			\begin{scope}[shift={(0,-3.5)}]
				\draw[interval,thick] (0,0) -- (12,0);
			\foreach [count=\j] \f/\t/\ignored in \nodes {
				\draw[interval,thick] (\f,.1) -- ++(0,-0.2);
				\draw[dashed,interval,thin] (0.5*\f+0.5*\t,0) coordinate (II\j) to (I\j) ;
			}
			\draw[interval,thick] (12,.1) -- ++(0,-0.2);
			\end{scope}
			\foreach \p/\c in {4/2,2/1,2/3,4/6,6/5,6/7,7/8} {
				\draw[interval,very thick] (I\p) -- (I\c) ;
			}
			\foreach [count=\j] \x in {-.7,{0.5*1.5-0.5*0.25},2.25,5.5,8.5,10.25,11.5,12.7} {
				\draw[interval] (\x,-4.25) node[fill=white,inner sep=0pt] {\small$I_{2,\j}$} [cyan!80!blue,dotted] to[out=90,in=-90] (II\j);
			}
			\coordinate (ll) at (-1.5,-4.75) ;
			\coordinate (ur) at (13.5,.4) ;
			\begin{scope}[transparency group,opacity=.4]
				\node[interval,draw,line width=3pt,rounded corners=20pt,fit=(ll) (ur)] (box) {} ;
				\draw[interval,ultra thick,densely dashed,-{Stealth[]}] (D2) to[out=45,in=90] (box.90) ;
				\draw[interval,dotted] (r1.90) -- (box.145) (r2.90) -- (box.35) ;
			\end{scope}
		\end{scope}

	\end{tikzpicture}
	\caption{Schematic view of the original lazy search tree data structure~\cite{SandlundWild20}.}
	\label{fig:lst}
\end{figure}

From the theory of efficient priority queues, it is well known that much more efficient implementations
are possible if not all of the above operations have to be supported: When only minimum, insert, delete, decrease-key, and meld are allowed, all can be supported to run in constant time, except for $\Oh{\log N}$ delete.
Lazy search trees~\cite{SandlundWild20,SandlundZhang22} (LSTs) show that we can get the same result with full support for all sorted-dictionary operations, so long as they are not used.  When they do get used, lazy search trees give the optimal guarantees of any comparison-based data structure for the given query ranks, gracefully degrading from priority-queue performance to BST performance, as more and more queries are asked. They therefore serve as a versatile drop-in replacement for both priority queues and BSTs.%
\footnote{
	We point out that a related line of work on ``deferred data structures'' also adapts to the queried ranks
	in the way LSTs do, but all data structures prior to LSTs had $\OhMega{\log N}$ insertion time (with the assumption that any insertion is preceded by a query for said element).  
	These data structures are therefore inherently unable to function as an efficient priority queue.
    Further related work is discussed in \wref{sec:related-work}.
}

The arguably most impactful application of sorted dictionaries across time is the database index.
Here, economizing on accesses to slow secondary storage can outweigh any other costs, which gave rise to external-memory data structures; in particular B-trees and their numerous variants. 
They achieve a $\log_2(B)$-factor speedup over BSTs in terms of input/output operations (I/Os), for $B$ the block size of transfers to the external memory, for scanning-dominated operations such as a range query, even a factor $B$ speedup can be achieved.

Much ingenuity has been devoted to (a) write-optimized variants of B-trees~\cite{Arge1995,BrodalFagerberg03}, often with batched queries, and (b), adaptive indices~\cite{IdreosManegoldGraefe2012}, or more recently learned indices~\cite{KraskaBeutelChiDeanPolyzotis2018} that try to improve query performance via adapting to data at hand.
\emph{Deferred data structuring}~-- known as \emph{database cracking} in the systems community~-- is a special case of (b) that, like lazy search trees, refines a sorted order (of an index) successively upon queries.
Yet in the past, these two approaches often remained incompatible, if not having outright conflicting goals,
making either insertions or queries fast.
Lazy search trees are a principled solution to get both (as detailed in \wref{sec:interface});
unfortunately, standard lazy search trees are not competitive with B-trees in terms of their I/O efficiency, thus losing their adaptive advantages when data resides on secondary storage.

In this paper, we present \emph{lazy B-trees}, an I/O-efficient variant of the original lazy search trees~\cite{SandlundWild20}.%
\footnote{
	The follow-up work~\cite{SandlundZhang22} replaces the second and third layer using \emph{selectable priority queues}; we discuss why these are challenging to make external in \wref{sec:contribution}.
}
Lazy search trees consist of three layers (see \wref{fig:lst}).
The topmost layer consists of a biased search tree of \emph{gaps} (defined in \wref{sec:lazy-search-trees}), weighted by their size.
The second layer consists, for each gap $\Delta_i$, of a balanced binary search tree of $\Oh{\log \lrSize{\Delta_i}}$ \emph{intervals} $I_{i,j}$.  Intervals are separated by splitter elements (pivots as in Quicksort),
but are unsorted within.
The third layer represents intervals as a simple unordered list.

A key obstacle to external-memory variants of lazy search trees~-- both for the original version~\cite{SandlundWild20} and for optimal lazy search trees~\cite{SandlundZhang22}~-- is the topmost layer.
As discussed in \wref{sec:related-work}, none of the existing data structures fully solves the problem of how to design a general-purpose \emph{external-memory biased search tree}~-- and we likewise leave this as an open problem for future work.
However, we show that for a slightly restricted set of performance guarantees sufficient for lazy search trees,
a (to our knowledge) novel construction based on partitioning elements by \emph{weight} with doubly-exponentially increasing bucket sizes provides an I/O-efficient solution.
This forms the basis of our lazy B-trees, which finally bring the adaptive versatility of lazy search trees to external memory.

\subsection{The External-Memory Model}
\label{sec:em_model}

In this paper we study the sorted dictionary problem in a hierarchical-memory model, where we have an unbounded external memory and an internal memory of capacity $M$ elements, and where data is transferred between the internal and external memory in blocks of $B$ consecutive elements. A block transfer is called an \emph{I/O} (input/output operation). The I/O cost of an algorithm is the number of I/Os it  performs. Aggarwal and Vitter~\cite{AggarwalVitter88} introduced this as the \emph{external-memory model} (and proved, e.g., lower bounds for sorting in the model).

\subsection{Lazy Search Trees}
\label{sec:lazy-search-trees}
\label{sec:interface}

In this section we briefly describe lazy search trees~\cite{SandlundWild20}.
Lazy search trees support all operations of a dynamic sorted dictionary (details in~\cite{SandlundWild20}).
We consider the following operations (sufficient to base a full implementation on):

\begin{description}
    \item[$\Construct(S)$:]
        Constructs the data structure from the elements of set $S$.%
        \footnote{\Construct and iterative insertion have the same complexity in internal memory, but in external memory, the bulk operation can be supported more efficiently.}
    \item[$\Insert(e)$:]
        Adds element $e$ to the set.
    \item[$\Delete(\mathit{ptr})$:]
        Deletes the element at pointer $\mathit{ptr}$ from the set.
    \item[$\ChangeKey(\mathit{ptr}, e')$:]
        Changes the element $e$ at pointer $\mathit{ptr}$ to $e'$.
    \item[$\QueryElement(e)$:]
        Locates the predecessor $e'$ to element $e$ in the set, and returns rank $r$ of $e'$, and a pointer to $e'$.
    \item[$\QueryRank(r)$:]
        Locates the element $e$ with rank $r$ in the set, and returns a pointer to $e$.
\end{description}

Lazy search trees maintain the elements of the set partitioned into $G$ \emph{``gaps''}, $\Delta_1,\ldots,\Delta_G$.
All elements in $\Delta_i$ are weakly smaller than all elements in $\Delta_{i+1}$, but the gaps are otherwise treated as unsorted bags.
Boundaries between gaps (and new gaps) are introduced only via a query; initially we have a single gap $\Delta_1$. 
Every query is associated with its \emph{query rank}~$r$~\cite[\S4]{SandlundWild20arxiv}, i.e., the rank of its result,
and splits an existing gap into two new gaps, such that its query rank is the largest rank in the left one of these gaps.
After queries with ranks $r_1<r_2<\cdots<r_q$, we thus obtain the gaps $\Delta_1,\ldots,\Delta_{q+1}$ where $|\Delta_i| = r_i-r_{i-1}$ upon setting $r_0=0$ and $r_{q+1}=n$.

Lazy search trees support insertions landing in gap $\Delta_i$ in $\Oh{\log \lrParens{N/\lrSize{\Delta_i}} + \log \log \lrSize{\Delta_i}}$ (worst-case) time and queries in $\Oh{x \log c + \log n}$ amortized time, where the query splits a gap into new gaps of sizes $x$ and $cx$ for some $c\ge1$.
Deletions are supported in $\Oh{\log n}$ (worst-case) time.
To achieve these times, gaps are further partitioned into intervals, with an amortized splitting and merging scheme (see \wref{sec:interval_structure} for details).

\subsection{Related Work}
\label{sec:related-work}

We survey key prior work in the following;
for a quick overview with selected results and comparison with our new results, see also \wref{tab:overview_of_results}.

\begin{table}[tbp]
    \centering
    
    \caption{Overview of results on Search Trees and Priority Queues for both the internal- and external-memory model. For the internal model the displayed time is the number of RAM operations performed, while for the external model the displayed time is the number of \IOs performed. Amortized times are marked by ``am.'' and results doing batched updates and queries are marked by ``batched''. For priority queues, query is only for the minimum. All results use linear space.}
    \label{tab:overview_of_results}
    
    \newcommand{\lines}[2]{\begin{tabular}{@{}l@{}}#1\\#2\end{tabular}}
    \small
	\setlength{\extrarowheight}{2pt}
    \begin{tabular}{l@{}cc@{}}
    	\toprule
         & \textbf{Insert} & \textbf{Query} \\ 
         \midrule
        \textbf{Internal-Memory} \\
        Balanced BST~\cite{AVL62,GuibasSedgewick78}
        & $\Oh{\log N}$
        & $\Oh{\log N}$ \\
        Priority Queue~\cite{Brodal1996,BrodalLagogiannisTarjan2012}
        & $\Oh{1}$
        & $\Oh{\log N}$ \\
        Lazy Search Tree~\cite{SandlundWild20}
        & $\Oh[\big]{\log \frac{N}{\lrSize{\Delta_i}} + \log \log N}$
        & $\Oh{\log N + x \log c}$ am.\\
        Optimal LST~\cite{SandlundZhang22}
        & $\Oh[\big]{\log \frac{N}{\lrSize{\Delta_i}}}$
        & $\Oh{\log N + x \log c}$ am.
        \\[1ex] 
        \midrule
        \textbf{External-Memory} \\
        B-tree~\cite{BayerMcCreight72}
        & $\Oh{\logB N}$ 
        & $\Oh{\logB N}$ \\
        B$^{\varepsilon}$-tree~\cite{BrodalFagerberg03}
        & $\Oh{\frac{1}{\varepsilon B^{1-\varepsilon}}\logB N}$ am.
        & $\Oh{\frac{1}{\varepsilon}\logB N}$ am. \\
        Buffer Tree~\cite{Arge1995,Arge03}
        & $\Oh{\frac 1B \log_{M/B} \frac{N}{B}}$ batched
        & $\Oh{\frac 1B \log_{M/B} \frac{N}{B}}$ batched \\
        I/O-eff. heap~\cite{KumarS96}
        & $\Oh{\frac1B \log_2 \frac{N}{B}}$ am.
        & $\Oh{\frac1B \log_2 \frac{N}{B}}$ am.\\
        $x$-treap heap~\cite{IaconoJacobTsakalidis2019}
        & $\Oh{\frac1B \log_{M/B} \frac{N}{B}}$ am.
        & $\Oh{\frac{M^\varepsilon}{B} \log^2_{M/B} \frac{N}{B}}$ am. \\[3pt]
        \makecell[l]{\emph{This paper, LST}\\\quad\emph{(\wref{thm:main_structure}})}
        & $\Oh[\big]{\logB \frac{N}{\lrSize{\Delta_i}} + \logB \logB \lrSize{\Delta_i}}$ 
        & \makecell[l]{$\mathcal{O}\big(\logB \min \left\{ N, q \right\} + \frac{1}{B} \log_2 \lrSize{\Delta_i}$\\\quad $ + \logB \logB \lrSize{\Delta_i} + \frac{1}{B} x \log_2 c\big)$ am.}\\[6pt]
        \makecell[l]{\emph{This paper, PQ}\\\quad\emph{(\wref{cor:priority_queue})}}
        & $\Oh{\logB \logB N}$
        & $\Oh{\frac{1}{B} \log_2 N + \logB \logB N}$ am. \\
        \bottomrule
    \end{tabular}
\end{table}

\subparagraph{(External) Search trees.}

Balanced binary search trees (BSTs) exist in many varieties, with AVL trees~\cite{AVL62} and red-black trees~\cite{GuibasSedgewick78} the most widely known. When augmented with subtree sizes,
they support all sorted dictionary operations in $\Oh{\log N}$ worst-case time.
Simpler variants can achieve the same via randomization~\cite{SeidelAragon96,MartinezRoura1998} or amortization~\cite{SleatorTarjan85}.
In external memory, B-trees~\cite{BayerMcCreight72}, often in the leaf-oriented flavor as B${}^+$-trees and augmented with subtree sizes,
are the benchmark.  They support all sorted-dictionary operations in $\Oh{\logB N}$ \IOs.
By batching queries, buffer trees~\cite{Arge1995,Arge03} substantially reduce the cost to amortized
$\Oh[\big]{\frac 1B \log_{M/B} \frac{N}{B}}$ \IOs, but are mostly useful in an offline setting due to the long delay from query to answer.
$\text B^\varepsilon$-trees~\cite{BrodalFagerberg03} avoid this with a smaller buffer of operations per node to achieve amortized $\Oh{\frac{1}{\varepsilon B^{1-\varepsilon}}\logB N}$ \IOs for updates and $\Oh{\frac{1}{\varepsilon}\logB N}$ \IOs for queries (with immediate answers), where $\varepsilon\in(0,1]$ is a parameter.

\subparagraph{Dynamic Optimality.}

The dynamic-optimality conjecture for Splay trees~\cite{SleatorTarjan85} resp. the GreedyBST~\cite{Lucas1988,Munro2000,DemaineHarmonIaconoKanePatrascu2009} algorithm posits that these methods provide an instance-optimal binary-search-tree algorithm for any (long) sequence of searches over a static set of keys in a binary search tree.
While still open for Splay trees and GreedyBST, the dynamic optimality of data structures has been settled in some other models:
it holds for a variant of skip lists~\cite{BoseDouiebLangerman2008}, multi-way branching search trees, and B-trees~\cite{BoseDouiebLangerman2008};
it has been refuted for tournament heaps~\cite{MunroPengWildZhang2019}.
As in lazy search trees, queries clustered in time or space allow a sequence to be served faster.
Unlike in lazy search trees, insertions can remain expensive even when no queries ever happen close to them.
A more detailed technical discussion of similarities and vital differences compared to lazy search trees appears in~\cite{SandlundWild20,SandlundWild20arxiv}.

\subparagraph{(External) Priority Queues.}

When only minimum-queries are allowed, a sorted dictionary becomes a priority queue (PQ) (a.k.a.\ ``heap'').
In internal memory, all operations except delete-min can then be supported in $\Oh{1}$
amortized~\cite{FredmanTarjan1987} or even worst-case time~\cite{Brodal1996,BrodalLagogiannisTarjan2012}.
In external memory, developing efficient PQs has a long history.
For batch use, buffer-tree PQs~\cite{Arge03} and the I/O-efficient heap of~\cite{KumarS96} support $k$ operations of insert and delete-min in $\Oh[\big]{\frac{k}{B} \log_{M/B} \frac{N}{M}}$ \IOs, with $N$ denoting the maximum number of stored keys over the $k$ operations.
The same cost per operation can be achieved in a worst-case sense~\cite{BrodalKatajainen98} (i.e., $B$ subsequent insert/delete-min operations cost $\Oh[\big]{\log_{M/B} \frac{N}{M}}$ \IOs).

None of these external-memory PQs supports decrease-key.
The I/O-efficient tournament trees of~\cite{KumarS96} support a sequence of insert, delete-min, and decrease-key
with $\Oh{\frac1B \log_2 \frac{N}{B}}$ \IOs per operation.
A further $\log\log N$ factor can be shaved off~\cite{JiangLarsen2019} (using randomization), but that, surprisingly, is optimal~\cite{EenbergLarsenYu2017}.
A different trade-off is possible: insert and decrease-key are possible in $\Oh[\big]{\frac1B \log_{M/B} \frac{N}{B}}$ amortized \IOs at the expense of driving up the cost for delete/delete-min to $\Oh[\big]{\frac{M^\varepsilon}{B} \log^2_{M/B} \frac{N}{B}}$~\cite{IaconoJacobTsakalidis2019}.

\subparagraph{(External) Biased Search Trees.}

Biased search trees~\cite{BentSleatorTarjan1985} maintain a sorted set, 
where each element $e$ has a (dynamic) weight $w(e)$, such that operations in the tree spend $\Oh[\big]{\log \frac{W}{w(e)}}$ time, for $W$ the total weight, $W=\sum_e w(e)$.
Biased search trees can be built from height-balanced trees ($(2,b)$-globally-biased trees~\cite{BentSleatorTarjan1985}) or
weight-balanced trees (by representing high-weight elements several times in the tree~\cite{Mehlhorn1984}), and
Splay trees automatically achieve the desired time in an amortized sense~\cite{SleatorTarjan85,Mehlhorn1984}.

None of the original designs are well suited for external memory.
Feigenbaum and Tarjan~\cite{FeigenbaumTarjan1983} extended biased search trees to $(a,b)$-trees for that purpose. However, during the maintenance of $(a,b)$-trees, some internal nodes may have a degree much smaller than $a$ (at least $2$), which means that instead of requiring $\Oh{N/B}$ blocks to store $N$ weighted elements, they require $\Oh N$ blocks in the worst case.%
\footnote{%
In particular, in \cite{FeigenbaumTarjan1983}, they distinguish internal nodes between minor and major, minor being the nodes that have degree $<a$ or have a small rank. All external nodes are major.%
}
The authors in~\cite{FeigenbaumTarjan1983} indeed leave it as an open problem to find a space-efficient version of biased $(a,b)$-trees.
Another attempt at an external biased search tree data structure is based on deterministic skip lists~\cite{BagchiBuchsbaumGoodrich2005}.
Yet again, the space usage seems to be $\OhMega{N}$ blocks of memory.%
\footnote{%
	Unfortunately, it remains rather unclear from the description in the article exactly which parts of the skiplist ``towers'' of pointers of an element are materialized.  Unlike in the unweighted case, an input could have large and small weight alternating, with half the elements of height $\approx \log_b W$.
	Fully materializing the towers would incur $\OhMega{n \log_b W}$ space; otherwise this seems to require a sophisticated scheme to materialize towers on demand, e.g., upon insertions, and we are not aware of a solution.%
}

To our knowledge, no data structure is known that achieves $\Oh[\big]{\logB \frac{W}{w(e)}}$ \IOs for a dynamic weighted set in external memory while using $\Oh{N/B}$ blocks of memory.

\subparagraph{Deferred Data Structures \& Database Cracking.}

Deferred data structuring refers to the idea of successively establishing a query-driven structure on a dataset. 
This is what lazy search trees do upon queries.
While the term is used more generally, it was initially proposed in~\cite{KarpMotwaniRaghavan1988} on the example of a (sorted) dictionary.  For an offline sequence of $q$ queries on a given (static) set of $N$ elements,
their data structure uses $\Oh{N \log q + q\log N}$ time.
In~\cite{ChingMehlhornSmid1990}, this is extended to allow update operations in the same time (where $q$ now counts all operations).
The refined complexity taking query \emph{ranks} into account was originally considered for the (offline) multiple selection problem: when $q$ ranks $r_1<\cdots<r_q$ are sought, leaving gaps $\Delta_1,\ldots,\Delta_{q+1}$,
$\Theta\bigl({\sum_{i=1}^{q+1}|\Delta_i|\log(N/\Delta_i)\bigr)}$ comparisons are necessary and sufficient~\cite{DobkinMunro1981,KaligosiMehlhornMunroSanders2005}.
For multiple selection in external memory, $\Theta\bigl({\sum_{i=1}^{q+1}\frac{|\Delta_i|}{B}\log_{M/B}\frac{N}{|\Delta_i|}}\bigr)$
\IOs are necessary and sufficient~\cite{BarbayGuptaRaoJon2016,BrodalWild2023}, even cache-obliviously~\cite{BrodalWild2023,BrodalWild2024}.

Closest to lazy search trees is the work on \emph{online dynamic multiple selection} by Barbay et al.~\cite{BarbayGuptaRaoJon2015,BarbayGuptaRaoJon2016},
where online refers to the query ranks arriving one by one.
As pointed out in~\cite{SandlundWild20}, the crucial difference between all these works and lazy search trees
is that the analysis of dynamic multiple selection assumes that every insertion is preceded by a query for the element, which implies that insertions must take $\Omega(\log N)$ time.
(They assume a nonempty set of elements to initialize the data structure with, for which no pre-insertion queries are performed.)
Barbay et al.~also consider online dynamic multiple selection in external memory.
By maintaining a B-tree of the pivots for partitioning, they can support updates~-- again, implicitly preceded by a query~-- at a cost of $\Theta(\logB N)$ \IOs each.

In the context of adaptive indexing of databases, deferred data structuring is known under the name
of \emph{database cracking}~\cite{IdreosKerstenManegold2007,Idreos2010,HalimIdreosKarrasYap2012}.
While the focus of research is on systems engineering, e.g., on the partitioning method~\cite{PirkPetrakiIdreosManegoldKersten2014}, some theoretical analyses of devised algorithms have also appeared~\cite{ZardbaniAfshaniKarras2020,Tao2025}.
These consider the worst case for $q$ queries on $N$ elements similar to the original works on deferred data structures.

\subsection{Contribution}
\label{sec:contribution}

Our main contribution, the lazy B-tree data structure, is summarized in \wref{thm:main_structure} below.

\begin{restatable}[Lazy B-Trees]{theorem}{thmMainStructure}
\label{thm:main_structure}
    There exists a data structure over an ordered set, that supports
    \begin{itemize}
        \item $\Construct(S)$ in worst-case $\Oh{\lrSize{S} / B}$ \IOs, 
        \item \Insert in worst-case $\Oh[\big]{\logB \frac{N}{\lrSize{\Delta_i}} + \logB \logB \lrSize{\Delta_i}}$ \IOs,
        \item \Delete in amortized $\Oh[\big]{\logB \frac{N}{\lrSize{\Delta_i}} + \frac{1}{B} \log_2 \lrSize{\Delta_i} + \logB \logB \lrSize{\Delta_i}}$ \IOs,
        \item \ChangeKey in worst-case $\Oh[\big]{\logB \logB \lrSize{\Delta_i}}$ \IOs if the element is moved towards the nearest queried element but not past it, and amortized $\Oh{\frac{1}{B} \log_2 \lrSize{\Delta_i} + \logB \logB \lrSize{\Delta_i}}$ \IOs otherwise,
        and 
        \item \QueryElement and \QueryRank in amortized\\ 
        	$\Oh[\big]{\logB \min \left\{ N, q \right\} + \frac{1}{B} \log_2 \lrSize{\Delta_i} + \logB \logB \lrSize{\Delta_i} + \frac{1}{B} x \log_2 c}$ \IOs.
    \end{itemize}
    Here $N$ denotes the size of the current set, $\lrSize{\Delta_i}$ denotes the size of the manipulated gap, $q$ denotes the number of performed queries, and $x$ and $cx$ for $c \ge 1$ are the resulting sizes of the two gaps produced by a query.
    The space usage is $\Oh{N / B}$ blocks.
    All results ignore the cost of keeping external pointers to elements in the data structure up to date.
\end{restatable}

From the above theorem, the following corollary can be derived, which states the performance of lazy B-trees when used as a priority queue. 

\begin{restatable}[Lazy B-Trees as external PQ]{corollary}{corPriorityQueue}
\label{cor:priority_queue}
    A lazy B-tree may be used as a priority queue, to support, within $\Oh{N / B}$ blocks of space,
    the operations
    \begin{itemize}
        \item $\Construct(S)$ in worst-case $\Oh{\lrSize{S} / B}$ \IOs,
        \item \Insert in worst-case $\Oh{\logB \logB N}$ \IOs,
        \item \Delete in amortized $\Oh{\frac{1}{B} \log_2 N + \logB \logB N}$ \IOs,
        \item \DecreaseKey in worst-case $\Oh{\logB \logB N}$ \IOs, and 
        \item \Minimum in amortized $\Oh{\frac{1}{B} \log_2 N + \logB \logB N}$ \IOs.
    \end{itemize}
    All results ignore the cost of keeping external pointers to elements in the data structure up to date.
\end{restatable}

\medskip

The running times of lazy B-trees when used as a priority queue are not competitive with heaps 
targeting sorting complexity (such as buffer trees~\cite{Arge1995,Arge03}); 
however these data structures do not (and cannot~\cite{EenbergLarsenYu2017}) support decrease-key efficiently.
By contrast, for very large $N$, lazy B-trees offer exponentially faster decrease-key and insert than previously known external priority queues, while only moderately slowing down delete-min queries.

Our key technical contribution is a novel technique for partially supporting external biased search tree performance, formally stated in \wref{thm:gap_structure} below. 
In the language of a biased search tree, it supports 
searches (by value or rank)
as well as
incrementing or decrementing%
\footnote{%
    General weight changes are possible in that time with $w(e)$ the minimum of the old and new weight.
}
a weight $w(e)$ by $1$
in $\Oh{\logB(W/w(e))}$ \IOs for an element $e$ or weight $w(e)$;
inserting or deleting an element, however, takes $\Oh{\logB N}$ \IOs irrespective of weight, where $N$ is the number of elements currently stored.
Unlike previous approaches, the space usage is the $\Oh{N/B}$ blocks throughout.
A second technical contribution is the streamlined potential-function-based analysis of 
the interval data structure of lazy search trees.

We mention three, as yet insurmountable, shortcomings of lazy B-trees.
The first one is the $\log \log N$ term we inherit from the original Lazy search trees~\cite{SandlundWild20}. This cost term is in addition to the multiple-selection lower bound and thus not necessary.  Indeed, it was in internal memory subsequently removed~\cite{SandlundZhang22}, using an entirely different representation of gaps,
which fundamentally relies on soft-heap-based selection on a priority queue~\cite{KaplanKozmaZamirZwick2019}.
The route to an external memory version of this construction is currently obstructed by two road blocks.
First, we need an external-memory soft heap; the only known result in this direction~\cite{BhushanGopalan2012}
only gives performance guarantees when $N = \Oh[\big]{B (M/B)^{M/2(B+\sqrt{M/B})}}$ and hence seems not to represent a solution for the general problem.
Second, the selection algorithm from~\cite{SandlundZhang22} requires further properties of the priority queue implementation, in particular a bound on the fanout; it is not clear how to combine this with known external priority queues.

The second shortcoming is that~-- unlike for comparisons~-- we do not get close to the \IO-lower bound for multiple-selection with lazy B-trees.  
Doing so seems to require a way of buffering as in buffer trees, to replace our fanout of $B$ by a fanout of $M/B$.
This again seems hard to achieve since an \emph{insertion} in the lazy search tree uses a \emph{query} on the gap data structure, and a \emph{query} on the lazy search tree entails an \emph{insertion} into the gap data structure (plus a re-weighting operation). 

Third, the above running times do \emph{not} include the cost to keep external pointers to elements (as used for \Delete{} or \DecreaseKey) up to date.
When splitting or fusing nodes of the blocked-linked list representing intervals (as described in detail in \wref{app:interval_structure_full_version}), we have to update pointers to elements that have moved between blocks. Since these external pointers may reside scattered across memory, we may have to pay one \IO per pointer, which increases the worst-case cost contribution $\frac1B x \log_2 c$ in queries to $x \log_2 c$.

\subparagraph{Outline.}

The remainder of the paper is structured as follows.
\wref{sec:gap_structure} describes the gap data structure (our partial external biased search tree)
and key innovation.
\wref{sec:interval_structure} sketches the changes needed to turn the interval data structure from~\cite{SandlundWild20} into an \IO-efficient data structure;
the full details, including our streamlined potential function, appear in \wref{app:interval_structure_full_version}.
In \wref{sec:combined_structure}, we then show how to assemble the pieces into a proof of our main result, \wref{thm:main_structure}.
We conclude in \wref{sec:open_problems} with some open problems.
To be self-contained, we include proofs of some technical lemmas used in the analysis in \wref{app:support_lemma}.

\section{The Gap Structure: A New Restricted External Biased Search Tree}
\label{sec:gap_structure}

In this section we present a structure on the gaps, which allows for the following operations.
Let $N$ denote the total number of elements over all gaps and $G$ denote the number of gaps. Note that $G \le N$, as no gaps are empty.
We let $\Delta_i$ denote the $i$th gap when sorted by element, and $\lrSize{\Delta_i}$ denote the number of elements contained in gap $\Delta_i$.

\begin{description}
    \item[$\SearchGap(e)$]
      Locates the gap $\Delta_i$ containing element $e$.
    \item[$\RankGap(r)$]
      Locates the gap $\Delta_i$ containing the element of rank $r$.
    \item[$\IncrementGap(\Delta_i)$ / $\DecrementGap(\Delta_i)$]
      Changes the weight of gap $\Delta_i$ by $1$.
    \item[$\SplitGap(\Delta_i, \Delta'_i, \Delta'_{i + 1})$]
     Splits gap $\Delta_i$ into the two non-overlapping gaps $\Delta'_i$ and $\Delta'_{i + 1}$, s.t. the elements of $\Delta_i$ are equal to the elements of $\Delta'_i$ and $\Delta'_{i + 1}$.
\end{description}

For the operations, we obtain \IO costs, as described in the theorem below.

\begin{theorem}[Gap Data Structure]
\label{thm:gap_structure}
    There exists a data structure on a weighted ordered set, that supports \SearchGap, \RankGap, \IncrementGap and \DecrementGap in $\Oh[\big]{\logB \frac{N}{\lrSize{\Delta_i}}}$ \IOs, and \SplitGap in $\Oh{\logB G}$ \IOs.
    Here $N$ denotes the total size of all gaps, $\lrSize{\Delta_i}$ denotes the size of the touched gap and $G$ denotes the total number of gaps.
    The space usage is $\Oh{G / B}$ blocks.
\end{theorem}

\begin{remark}[Comparison with biased search trees]
A possible solution would be an external-memory version of biased search trees,
but as discussed in the introduction, no fully satisfactory such data structure is known.
Instead of supporting all operations of biased trees in full generality, we here opt for a solution solving (just) the case at hand.
The solution falls short of a general biased search trees, as the insertion or deletion costs are not a function of the \emph{weight} of the affected gap, but the \emph{total number} $G$ of gaps in the structure.
Moreover, we only describe how to change the weight of gaps by $1$; however, general weight changes could be supported, with the size of the change entering the \IO cost, matching what we expect from a general biased search tree.
\end{remark}

Note that the gaps in the structure are non-overlapping, and that their union covers the whole element range. The query for an element contained in some gap is therefore a predecessor query on the left side of the gaps, however, as their union covers the entire element range, the queries on the gap structure behave like \emph{exact queries}: we can detect whether we have found the correct gap and can then terminate the search.
(By definition, there are no unsuccessful searches in our scenario, either.)

For consistency with the notation of biased search trees, we write in the following $w_i = \lrSize{\Delta_i}$ and $W = \sum_i w_i$. Note that $W = N$.
Consider a conceptual list, where the gaps are sorted decreasingly by \emph{weight}, and let gap~$\Delta_i$ be located at some index $\ell$ in this list. This conceptual list is illustrated in \wref{fig:conceptual_list_and_splits}, and \wref{fig:two_dimentional_gaps} gives a two-dimensional view of the list, which considers both the gap weight as well as the gaps ordered by element. As the total weight before index $\ell$ in the conceptual list is at most $W$ and the weight of each gap before index~$\ell$ is at least $w_i$, then it must hold that~$\ell \le \frac{W}{w_i}$.
If we were to search for a gap of known \emph{weight}, it can therefore be found with an exponential search~\cite{BentleyYao76} in time $\Oh{\log \ell} = \Oh[\big]{\log \frac{W}{w_i}}$. 
However, searches are based on \emph{element values} (and, e.g., for insertions, without knowing the target gap's weight), so this alone would not work.

\begin{figure}[tb]
    \centering
    
    \begin{tikzpicture}[scale=0.7]
        \foreach \i in {0, ..., 16} \draw ($(\i, 0) + (0.5, 0)$) -- ++(0, 1);
        \draw (-0.5, 0) rectangle ++(18, 1);
        \draw (17.5, 0) -- ++(0.5,0) -- ++(-0.25,0.25) -- ++(0.25,0.25) -- ++(-0.25,0.25) -- ++(0.25,0.25) -- ++(-0.5,0);

        \foreach \i in {0, ..., 8} \node at (\i, -0.3) {\tiny $\i$};
        \node at (9, -0.3) {\tiny $\dots$};
        \node at (10, -0.3) {\tiny $\ell$};
        \node at (11, -0.3) {\tiny $\dots$};
        
        \foreach \i [count=\j from 0] in {4, 9, 7, 2, 8, 6, 3, 5, 1} \node at (\j, 0.5) {\small $\Delta_\i$};
        \node at (10, 0.5) {\small $\Delta_i$};

        \foreach \i in {1.5, 5.5} \draw[densely dashed, very thick, gray] (\i, -0.5) -- (\i, 1.8);
        \foreach \i [count=\j from 0] in {0.5, 3.5, 11.5} \node[gray] at (\i, 1.5) {\small $b_\j$};
        \node[gray] at (-.8,1.5) {\footnotesize \emph{early}} ;
        \node[gray] at (17.5,1.5) {\footnotesize \emph{late buckets}} ;
    \end{tikzpicture}

    \caption{The conceptual list of the gaps. The gaps are sorted by decreasing weight, with the heaviest gap (largest weight), $\Delta_4$ at index $0$.
    Gap $\Delta_i$ with weight $w_i$ is stored at index $\ell \le W/w_i$ in the list.
    The gaps are split into buckets $b_0, b_1, b_2, \ldots$ of doubly exponential size.}
    \label{fig:conceptual_list_and_splits}
\end{figure}
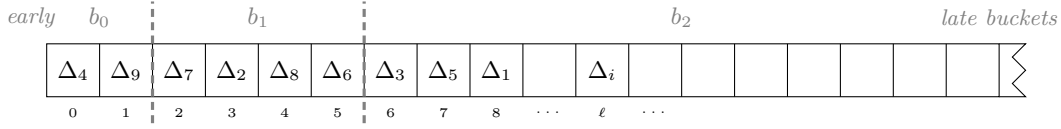

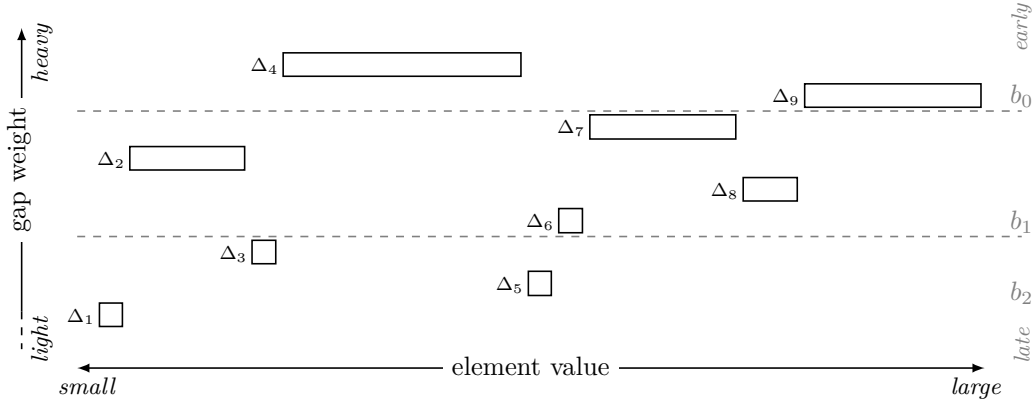
\begin{figure}[tb]
    \centering
    
    \begin{tikzpicture}[xscale=.98,semithick]
        \draw[latex-latex] (0,-0.5) -- +(12.25,0) node[midway, fill=white, yshift=1pt] {element value}
        	node[pos=0.01,below] {\small \itshape small}
        	node[pos=.99,below] {\small \itshape large}
        	;
        \draw[dashed] (-0.75,0.25) -- +(0,-0.5);
        \draw[-latex] (-0.75,0.25) -- +(0,3.75) node[midway, fill=white, rotate=90, yshift=1pt] {gap weight}
              node[pos=-0.1,below,sloped] {\small \itshape light}
              node[pos=.91,below,sloped] {\small \itshape heavy}
        ;
		\begin{scope}[shift={(.25,0)}]
            \def\margin{0.05}
            \foreach \i/\l/\r/\t/\b in {4/2.48/5.79/3.72/3.31, 9/9.52/12.0/3.31/2.9, 7/6.62/8.69/2.9/2.48, 2/0.41/2.06/2.48/2.07, 8/8.69/9.52/2.07/1.66, 6/6.2/6.62/1.66/1.24, 3/2.06/2.48/1.24/0.83, 5/5.79/6.2/0.83/0.41, 1/0.0/0.41/0.41/0.0} {
                \draw (\l+\margin, \t-\margin) node[xshift=-7pt, yshift=-5pt] {\scriptsize $\Delta_{\i}$} rectangle (\r-\margin, \b+\margin);
            }
    
            \draw[dashed, gray] (-0.25, 2.9) -- (12.5, 2.9);
            \draw[dashed, gray] (-0.25, 1.24) -- (12.5, 1.24);
            \node[gray,rotate=90] at (12.5,4) {\footnotesize\itshape early} ;
            \node[gray] at (12.5, 3.1) {$b_0$};
            \node[gray] at (12.5, 1.44) {$b_1$};
            \node[gray] at (12.5, 0.5) {$b_2$};
            \node[gray,rotate=90] at (12.5,-.2) {\footnotesize\itshape late} ;
        \end{scope}
    \end{tikzpicture}
    
    \caption{Two dimensional view of the first gaps of \wref{fig:conceptual_list_and_splits} represented by rectangles, with the width of the rectangle denoting the weight (size) of the gap. The horizontal axis sorts gaps by element value; the vertical axis by weight. Dashed lines show bucket boundaries.
    }
    \label{fig:two_dimentional_gaps}
\end{figure}

\subsection{Buckets by Weight and Searching by Element}

Instead, the gaps are split into buckets $b_j$, s.t.\ the weight of all gaps in bucket $b_j$ is greater than or equal to the weight of all gaps in bucket $b_{j + 1}$.
We further impose the invariant that the size of (number of gaps in) bucket $b_j$ is \emph{exactly} $\smash{B^{2^j}}$ for all but the last bucket, which may be of any (smaller) size.
Each bucket is a B-tree containing the gaps in the bucket, sorted by element value.
The time to search for a given element in bucket $b_j$ is therefore $\Oh{\logB(|b_j|)} = \Oh{2^j}$ \IOs.
For consistent wording, we will in the following always use \emph{smaller/larger} to refer to the order by element values,
\emph{lighter/heavier} for gaps of smaller/larger weight/size, and \emph{earlier/later} for buckets of smaller/larger index.
Note that earlier buckets contain fewer, but heavier gaps.

\subparagraph{\RankGap.}

A search for a gap proceeds by searching in buckets $b_0, b_1, b_2, \ldots$ until the desired gap is found in some bucket $b_k$. To search in all buckets up until $b_k$ requires $\sum_{j=0}^k \Oh{2^j} = \Oh{2^k}$ \IOs, which is therefore up to constant factors the same cost as searching in only bucket $b_k$.
Consider some gap $\Delta_i$, which has the $\ell$th heaviest weight, i.e., it is located at index $\ell$, when sorting all gaps decreasingly by weight.
By the invariant, the buckets are sorted decreasingly by the weight of their internal elements. Let the bucket containing $\Delta_i$ be $b_k$. It must then hold that the sizes of the earlier buckets does not allow for $\Delta_i$ to be included, but that bucket $b_k$ does. Therefore,
\[ \sum_{j=0}^{k-1} \lrSize{b_j} \;<\; \ell \;\le\; \sum_{j=0}^{k} \lrSize{b_j} \; . \]
As $\lrSize{b_j} = B^{2^j}$, the sums are asymptotically equal to the last term of the sum up to constant factors (\wref{lem:theta_double_exp_sum} in \wref{app:support_lemma}).
It then holds that $\ell = \Oh[\big]{B^{2^k}}$ and $\ell = \OhMega[\big]{B^{2^{k-1}}} = \OhMega[\big]{(B^{2^k})^{1/2}}$.
Thus $\logB(\ell) = \OhTheta{\logB \lrSize{b_k}}$, and conversely $2^k = \OhTheta{\logB \ell}$.
The gap $\Delta_i$ can thus be found using $\Oh[\big]{\logB \frac{W}{w_i}}$~\IOs, concluding the \SearchGap operation.

\subsection{Updating the Weights}

Both explicit weight changes in \IncrementGap/\DecrementGap as well as the \SplitGap operation
require changes to the weights of gaps. Here, we have to maintain the sorting of buckets by weight.

\subparagraph{\IncrementGap/\DecrementGap.}
When performing \IncrementGap on gap $\Delta_i$, the weight is increased by $1$. In the conceptual list, this may result in the gap moving to an earlier index, with the order of the other gaps the same.
This move is made in the conceptual list (\wref{fig:conceptual_list_and_splits}) by swapping $\Delta_i$ with its neighbor. 
As the buckets are ordered by element value, swapping with a neighbor (by weight) in the same bucket does not change the bucket structure. When swapping with a neighbor in an earlier bucket, however, the bucket structure must change to reflect this.
Following the conceptual list (\wref{fig:conceptual_list_and_splits}), this gap is a lightest (minimum-weight) gap in bucket $b_{k-1}$.
This may then result in $\Delta_i$ moving out of its current bucket $b_k$, into some earlier bucket.
As the gap moves to an earlier bucket, some other gap must take its place in bucket $b_k$. Following the conceptual list, it holds that this gap is a lightest (minimum-weight) gap in bucket $b_{k-1}$, and so on for the remaining bucket moves.

When performing \DecrementGap, the gap may move to a later bucket, where the heaviest gap in bucket $b_{k+1}$ is moved to bucket $b_k$.

In both cases, a lightest gap in one bucket is swapped with a heaviest gap in the neighboring (later) bucket.
To find lightest or heaviest gaps in a bucket, we augment the nodes of the B-tree of each bucket with the values of the lightest and heaviest weights in the subtree for each of its children. This allows for finding a lightest or heaviest weight gap in bucket $b_j$ efficiently, using $\Oh{\logB{\lrSize{b_j}}}$ \IOs. Further, this augmentation can be maintained under insertions and deletions.

Upon a \IncrementGap or \DecrementGap operation, the desired gap $\Delta_i$ can be located in the structure using $\Oh[\big]{\logB \frac{W}{w_i}}$~\IOs, and the weight is adjusted.
We must then perform swaps between the buckets, until the invariant that buckets are sorted decreasingly by weight holds.
This swapping may be performed on all buckets up to the last touched bucket.
For \IncrementGap the swapping only applies to earlier buckets as the weight of $\Delta_i$ increases, and the height of the tree in the latest bucket is $\Oh[\big]{\logB \frac{W}{w_i}}$.
For \DecrementGap the swapping is performed into later buckets, but only up to the final landing place for weight $w_i-1$.
If $w_i \ge 2$, the height of the tree in the last touched bucket is $\Oh[\big]{\logB \frac{W}{w_i - 1}}=\Oh[\big]{1+\logB \frac{W}{w_i}}$.
If $w_i=1$, gap $\Delta_i$ is to be deleted. The gap is swapped until it is located in the last bucket, from where it is removed. We have $\Oh[\big]{\logB \frac{W}{w_i}} = \Oh[\big]{\logB W}$, as $w_i = 1$, whereas the height of the last bucket is $\Oh{\logB G} = \Oh{\logB W}$, as $G \le W$.
Therefore both \IncrementGap and \DecrementGap can be performed using $\Oh[\big]{\logB \frac{W}{w_i}}$~\IOs.

\subparagraph{\SplitGap.}

When \SplitGap is performed on gap $\Delta_i$, the gap is removed and replaced by two new gaps $\Delta'_i$ and $\Delta'_{i + 1}$, s.t.\ $\lrSize{\Delta_i} = \lrSize{\Delta'_i} + \lrSize{\Delta'_{\smash{i + 1}}}$. In the conceptual list, the new gaps must reside at later indexes than $\Delta_i$.
As the sizes of the buckets are fixed, and the number of gaps increases, some gap must move to the last bucket.
Similarly to \IncrementGap and \DecrementGap, the order can be maintained by performing swaps of gaps between consecutive buckets:
First, we replace $\Delta_i$ (in its spot) by $\Delta'_{i}$; if this violates the ordering of weights between buckets,
we swap $\Delta'_{i}$ with the later neighboring bucket, until the correct bucket is reached.
Then we insert $\Delta'_{i+1}$ into the last bucket and swap it with its earlier neighboring bucket until it, too, has reached its bucket.
Both processes touch at most all $\log_B \log_2 G$ buckets and spend a total of $\Oh{\logB G}$ \IOs.

\subsection{Supporting Rank Queries}

The final operation to support is \RankGap.
Let the \emph{global rank} of a gap $\Delta_i$, $r(\Delta_i) = |\Delta_1| + \cdots+ |\Delta_{i-1}|$, denote the rank of the smallest element located in the gap, i.e., the number of elements residing in gaps of smaller elements.
The \emph{local rank} of a gap $\Delta_i$ in bucket $b_j$, $r_j(\Delta_i)$, denotes the rank of the gap among elements located in bucket $b_j$ only.

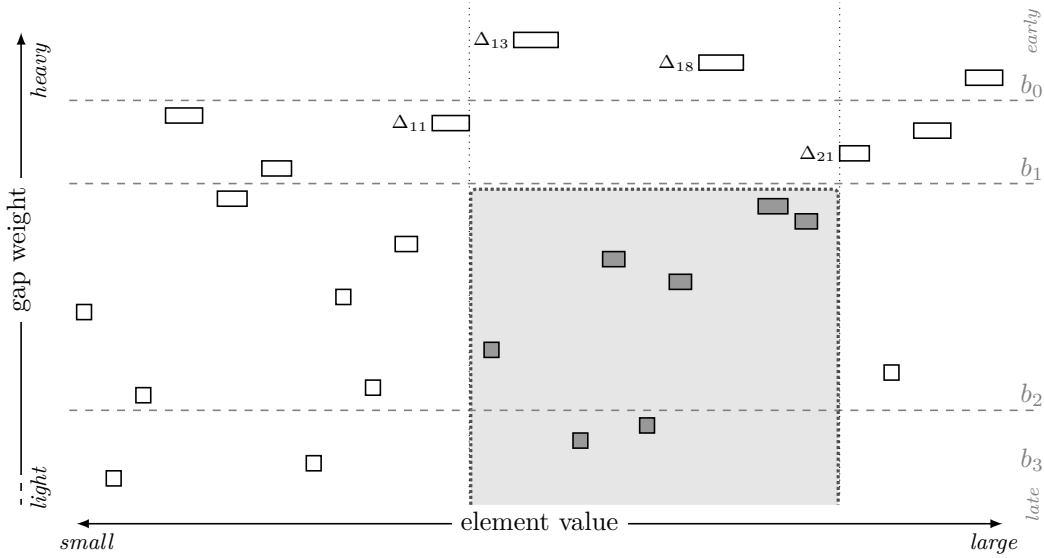
\begin{figure}[tb]
    \centering
    
    \begin{tikzpicture}[semithick,xscale=.98,yscale=1]
        \draw[latex-latex] (0,-0.5) -- +(12.5,0) node[midway, fill=white, yshift=1pt] {element value}
        	node[pos=0.01,below] {\footnotesize \itshape small}
        	node[pos=.99,below] {\footnotesize \itshape large}
        	;
        \draw[dashed] (-0.75,0.25) -- +(0,-0.5);
        \draw[-latex] (-0.75,0.25) -- +(0,5.75) node[midway, fill=white, rotate=90, yshift=1pt] {gap weight}
              node[pos=-0.05,below,sloped] {\footnotesize \itshape light}
              node[pos=.91,below,sloped] {\footnotesize \itshape heavy}
        ;
        
        \draw[darkgray, very thick,densely dotted,rounded corners=2pt,fill=black!10] (5.32, -0.25) -- (5.32, 3.925) -- (10.28, 3.925) -- (10.28, -0.25);

        \draw[dashed, gray] (-0.1, 5.1) -- +(13.1, 0);
        \draw[dashed, gray] (-0.1, 4) -- +(13.1, 0);
        \draw[dashed, gray] (-0.1, 1) -- +(13.1, 0);
        \node[gray] at (12.9, 5.3) {$b_0$};
        \node[gray] at (12.9, 4.2) {$b_1$};
        \node[gray] at (12.9, 1.2) {$b_2$};
        \node[gray] at (12.9, 0.35) {$b_3$};
        \node[gray,rotate=90] at (12.9,6) {\scriptsize\itshape early} ;
        \node[gray,rotate=90] at (12.9,-.25) {\scriptsize\itshape late} ;

        \foreach \x/\y/\w/\c in {0.0/2.2/0.2/0, 0.4/0.0/0.2/0, 0.8/1.1/0.2/0, 1.2/4.8/0.5/0, 1.9/3.7/0.4/0, 2.5/4.1/0.4/0, 3.1/0.2/0.2/0, 3.5/2.4/0.2/0, 3.9/1.2/0.2/0, 4.3/3.1/0.3/0, 4.8/4.7/0.5/0, 5.5/1.7/0.2/1, 5.9/5.8/0.6/0, 6.7/0.5/0.2/1, 7.1/2.9/0.3/1, 7.6/0.7/0.2/1, 8.0/2.6/0.3/1, 8.4/5.5/0.6/0, 9.2/3.6/0.4/1, 9.7/3.4/0.3/1, 10.3/4.3/0.4/0, 10.9/1.4/0.2/0, 11.3/4.6/0.5/0, 12.0/5.3/0.5/0} {
            \if \c 1
                \draw[fill=black!40] (\x, \y) rectangle ++(\w, 0.2);
            \else
                \draw (\x, \y) rectangle ++(\w, 0.2);
            \fi
        }

        \foreach \x/\y/\name in {4.8/4.7/11, 5.9/5.8/13, 8.4/5.5/18, 10.3/4.3/21} {
            \node at (\x-0.3,\y+0.1) {\scriptsize $\Delta_{\name}$};
        }

        \foreach \x in {5.3,10.3} {
            \draw[thin,dotted] (\x,-0.1) -- ++(0,6.5) ;
        }
    \end{tikzpicture}

    \caption{Illustration of ``holes''. The figure shows the two-dimensional view of gaps described in \wref{fig:two_dimentional_gaps} (gap widths not to scale).
    The \emph{hole} between $\Delta_{11}$ and $\Delta_{21}$ in bucket $b_1$, marked by the dotted line, contains the total weight of all gaps marked in gray.
    These are precisely the gaps that fall both (a) between $\Delta_{11}$ and $\Delta_{21}$ by element value and (b) in later buckets by (lighter) weight.
    Note specifically that $\Delta_{13}$ and $\Delta_{18}$ are \emph{not} counted in the hole, as they reside in an earlier bucket.
    }
    \label{fig:two_dimentional_with_hole}
\end{figure}

\subparagraph{Computing the rank of a gap.}

To determine the global rank of a gap $\Delta_i$, the total weight of all smaller gaps must be computed. This is equivalent to the sum of the local ranks of $\Delta_i$ over all buckets: $r(\Delta_i) = \sum_{j} r_j(\Delta_i)$.
First we augment the B-trees of the buckets, such that each node contains the total weight of all gaps in the subtree. 
This allows us to compute the local rank $r_j(\Delta_i)$ of a gap $\Delta_i$ inside a bucket $b_j$ using $\Oh{\logB \lrSize{b_j}}$ \IOs.  
The augmented values can be maintained under insertions and deletions in the tree within the stated \IO cost.
When searching for the gap $\Delta_i$, the local rank $r_j(\Delta_i)$ of all earlier buckets $b_j$ up to the bucket $b_k$, which contains $\Delta_i$, can then be computed using $\Oh[\big]{\logB \frac{W}{w_i}}$ \IOs in total.

It then remains to compute the total size of gaps smaller than $\Delta_i$ in all buckets after $b_k$, i.e., the sum of the local ranks $r_\ell(\Delta_i)$ for all later buckets $b_\ell$, $\ell > k$, to compute in total the global rank $r(\Delta_i)$. As these buckets are far bigger, we cannot afford to query them.
Note that any gaps in these later buckets must fall between two \emph{consecutive} gaps in~$b_k$, as the gaps are non overlapping in element-value space.
Denote the space between two gaps in a bucket as a \emph{hole}.
We then further augment the B-tree to contain in each hole of bucket $b_j$ the total size of gaps of \emph{later} buckets contained in that hole, and let each node contain the total size of all holes in the subtree.
See \wref{fig:two_dimentional_with_hole} for an illustration of which gaps contributes to a hole.

This then allows computing the global rank $r(\Delta_i)$ of gap $\Delta_i$, by first computing the local rank $r_j(\Delta_i)$ in all earlier buckets $b_j$, i.e., for all $j \le k$, and then adding the total size of all smaller holes in $b_k$.
The smaller holes in $b_k$ exactly sum to the local ranks $r_\ell(\Delta_i)$ for all later buckets $b_\ell$ for $\ell > k$.
As this in total computes $\sum_j r_j(\Delta_i)$, then by definition, we obtain the global rank $r(\Delta_i)$ of the gap $\Delta_i$.

\subparagraph{\RankGap.}

For a \RankGap operation, some rank $r$ is given, and the gap $\Delta_i$ is to be found, s.t. $\Delta_i$ contains the element of rank $r$, i.e., the global rank $r(\Delta_i) \le r < r(\Delta_{i+1})$.
The procedure is as follows.
First, the location of the queried rank $r$ is computed in the first bucket. If this location is contained in a gap $\Delta_i$ of the bucket, then $\Delta_i$ must be the correct gap, which is then returned.
Otherwise, the location of $r$ must be in a hole. As the sizes of the gaps contained in the first bucket is not reflected in the augmentation of later buckets, the value of $r$ is updated to reflect the missing gaps of the first bucket: we subtract from $t$ the local rank $r_0(\Delta_i)$ of the gap $\Delta_i$ immediately after the hole containing the queried rank $r$.
This adjusts the queried rank $r$ to be relative to elements in gaps of later buckets. 
Put differently, the initial queried rank $r$ is the queried global rank, which is the sum local ranks;
we now remove the first local rank again.
This step is recursively applied to later buckets, until the correct gap $\Delta_i$ is found in some bucket $b_k$, which contains the element of the initially queried global rank $r$.
As the correct gap $\Delta_i$ must be found in the bucket $b_k$ containing it, the \RankGap operation uses $\Oh[\big]{\logB \frac{W}{w_i}}$ \IOs to locate it. 

\subparagraph{Maintaining hole sizes.}

The augmentation storing the sizes of all holes in a subtree can be updated efficiently upon insertions or deletions of gaps in the tree, or upon updating the size of a hole.
However, computing the sizes of the holes upon updates in the tree is not trivial.
If a gap is removed from a bucket, the holes to the left resp.\ right of that gap are merged into a single hole, with a size equal to the sum of the previous holes.
If the removed gap is moved to a later bucket, the hole must now also include the size of the removed gap.
A newly inserted gap must, by the non-overlapping nature of gaps, land within a hole of the bucket, which is now split in two.
If the global rank of the inserted gap is known, then the sizes of the resulting two holes can be computed, s.t. the global rank is preserved.

In the operations on the gap structure, \SearchGap or \RankGap do not change the sizes of gaps, so the augmentation does not change.
Upon a \IncrementGap or \DecrementGap operation, the size of some gap $\Delta_i$ changes.
This change in size can be applied to all holes containing $\Delta_i$ in earlier buckets using $\Oh[\big]{\logB \frac{W}{w_i}}$ \IOs in total.
Then swaps are performed between neighboring buckets until the invariant that buckets are sorted by weight, holds again. 
During these swaps, the sizes of gaps do not change, which allows for computing the global rank of the gaps being moved, and update the augmentation without any overhead in the asymptotic number of \IOs performed. 
The total number of \IOs to perform a \IncrementGap or \DecrementGap operation does not increase asymptotically, and therefore remains $\Oh[\big]{\logB \frac{W}{w_i}}$.

When a \SplitGap is performed, a single gap $\Delta_i$ is split. 
As the elements in the two new gaps $\Delta'_i$ and $\Delta'_{i + 1}$ remain the same as those of the old gap $\Delta_i$, there cannot be another gap between them.
The value in all smaller holes therefore remains correct. Moving $\Delta'_{i+1}$ to the last bucket then only needs updating the value of the holes of the intermediate buckets, which touches at most all $\log_B \log_2 G$ buckets spending $\Oh{\logB G}$ \IOs in total. Swaps are then performed, where updating the augmentation does not change the asymptotic number of \IOs performed.

\subparagraph{Space Usage.}
To bound the space usage, note that the augmentation of the B-trees at most increase the node size by a constant factor. Since the space usage of a B-tree is linear in the number of stored elements, and since each gap is contained in a single bucket, the space usage is $\Oh{G / B}$ blocks in total.
This concludes the proof of \wref{thm:gap_structure}.

\section{The Interval Structure}
\label{sec:interval_structure}

A gap $\Delta_i$ contains all elements in the range the gap covers.
Note that by \wref{thm:main_structure},
a query on the gap for an element must be \emph{faster} for elements \emph{closer} to the border of a gap; information-theoretically speaking, such queries reveal less information about the data. 
It is therefore not sufficient to store the elements in the gap in a single list.

\subparagraph{External memory interval data structure.}
We follow the construction of Sandlund and Wild~\cite{SandlundWild20}, in their design of the interval structure. In this section we present the construction, and argue on the number of \IOs this construction uses.
The main difference from the original construction is in the speed-up provided by the external-memory model; namely that scanning a list is faster by a factor $B$, and that B-trees allows for more efficient searching. These differences are also what allows for slightly improving the overall \IO cost of the original construction, when moving it to the external-model.
Due to space constraints, the full analysis can be found in \wref{app:interval_structure_full_version}.

We allow for the following operations:

\begin{description}
    \item[$\InsertInterval(e)$:]
        Inserts element $e$ into the structure.
    \item[$\DeleteInterval(\mathit{ptr})$:]
        Deletes the element $e$ at pointer $\mathit{ptr}$ from the structure.
    \item[$\ChangeInterval(\mathit{ptr}, e')$:]
        Changes the element $e$ at pointer $\mathit{ptr}$ to the element $e'$.
    \item[$\SplitInterval(e)$ or $\SplitInterval(r)$:]
        Splits the set of intervals into two sets of intervals at element $e$ or rank $r$.
\end{description}

A gap has a \emph{``sidedness''}, which denotes the number of sides the gap has had a query. 
Denote a side as a \emph{queried side}, if that rank has been queried before (cf.\ \cite{SandlundWild20}).
If there have been no queries yet, the (single) gap is a $0$-sided gap.
When a query in a $0$-sided gap occurs, two new gaps are created which are both $1$-sided. 
Note that ``$1$-sided'' does not specify which side was queried~-- left or right.
When queries have touched both sides, the gap is $2$-sided.

We obtain the following \IO costs for the above operations.
\begin{theorem}[Interval Data Structure]
\label{thm:interval_structure}
    There exists a data structure on an ordered set, maintaining a set of intervals, supporting 
    \begin{itemize}
        \item \InsertInterval in worst-case $\Oh{\logB \logB \lrSize{\Delta_i}}$ \IOs, 
        \item \DeleteInterval in amortized $\Oh{\frac{1}{B} \log_2 \lrSize{\Delta_i} + \logB \logB \lrSize{\Delta_i}}$ \IOs, 
        \item \ChangeInterval in worst-case $\Oh{\logB \logB \lrSize{\Delta_i}}$ \IOs, if the element is moved towards the nearest queried side or amortized $\Oh{\frac{1}{B} \log_2 \lrSize{\Delta_i} + \logB \logB \lrSize{\Delta_i}}$ \IOs otherwise, and
        \item \SplitInterval in amortized $\Oh{\frac{1}{B} \log_2 \lrSize{\Delta_i} + \logB \logB \lrSize{\Delta_i} + \frac{1}{B} x \log_2 c}$ \IOs.
    \end{itemize}
    Here $\lrSize{\Delta_i}$ denotes the number of elements contained in all intervals, and $x$ and $cx$ for $c \ge 1$ are the resulting sizes of the two sets created by a split.
    The space usage is $\Oh{\lrSize{\Delta_i} / B}$~blocks.
\end{theorem}

\subparagraph{Intervals in external memory.}
Let the gap $\Delta_i$ contain multiple non-overlapping intervals $\Int{i, j}$, which contain the elements located in the gap.
The elements of the intervals are sorted between intervals, but not within an interval. Intervals therefore span a range of elements with known endpoints.
Each such interval is a blocked-linked-list containing the elements of the interval.
Additionally, we store an augmented B-tree over the intervals, allowing for efficiently locating the interval containing a given element (using $\Oh{\logB(\#\text{intervals})}$ \IOs).
The B-tree is augmented to hold in each node the total sizes of all intervals in the subtrees of the node, which allows for efficient rank queries.

By packing all intervals into a single blocked-linked-list, and noting that there can be no more intervals than there are elements, the space usage of the intervals and the augmented B-tree over the intervals is $\Oh{\lrSize{\Delta_i} / B}$ blocks.

\subparagraph{Intuition of interval maintenance of~\cite{SandlundWild20}.}
Intuitively, there is a trade-off in maintaining intervals: having many small intervals reduces future query costs since these are typically dominated by the linear cost of splitting one interval; but it increases the time to search for the correct interval upon insertions (and other operations).
Lazy search trees handle this trade-off as follows.
To bound the worst case insertion costs, we enforce a hard limit of $\Oh{\log_2 |\Delta_i|}$ on the number of intervals in a gap $\Delta_i$, implemented via a merging rule.
To amortize occasional high costs for queries, we accrue potential for any intervals that have grown ``too large'' relative to their proximity to a gap boundary.

Given the logarithmic number of intervals, the best case for queries would be to have interval sizes grow exponentially towards the middle. This makes processing of intervals close to the boundary (i.e., close to previous queries) cheap; for intervals close to the middle, we can afford query costs linear in $\Delta_i$.
It can be shown that for exponentially growing intervals, the increase of the lower bound from any query allows us to precisely pay for the incurred splitting cost.
However, the folklore rule of having each interval, say, 2--4 times bigger than the previous interval seems too rigid to maintain. Upon queries, it triggers too many changes.

\subparagraph{Merging \& Potential.}
We therefore allow intervals to grow bigger than they are supposed to be, but charge them for doing so in the potential.
By enforcing the merging rule when the number of elements in the gap decreases, we maintain the invariant that the number of intervals is $\Oh{\log_2 \lrSize{\Delta_i}}$, which allows for locating an interval using  $\Oh{\logB \log_2 \lrSize{\Delta_i}} = \Oh{\logB \logB \lrSize{\Delta_i}}$ \IOs.
Enforcing the merging rule is achieved by scanning the intervals using the B-tree, and it can be shown that this operation uses amortized $\Oh{\frac{1}{B} \log_2 \lrSize{\Delta_i}}$ \IOs. Merging intervals takes $\Oh{1}$ \IO, but by adding extra potential to intervals, this step of the merge can be amortized away.

Upon this, we can show that \InsertInterval uses $\Oh{\logB \logB \lrSize{\Delta_i}}$ \IOs, and that \DeleteInterval uses amortized $\Oh{\frac{1}{B} \log_2 \lrSize{\Delta_i} + \logB \logB \lrSize{\Delta_i}}$ \IOs.
When performing \ChangeInterval, moving the element from one interval to another uses $\Oh{\logB \logB \lrSize{\Delta_i}}$ \IOs. However, the potential function causes an increase of $\Oh{\frac{1}{B} \log_2 \lrSize{\Delta_i}}$~\IOs in the amortized cost when an element is moved away from the closest queried side of the gap.

When a query occurs, an interval $\Int{i, j}$ is split around some element or rank. The interval can be located using the B-tree over intervals on either element or rank in $\Oh{\logB \logB \lrSize{\Delta_i}}$ \IOs.
For splitting around an element, the interval is scanned and partitioned using $\Oh{\lrSize{\Int{i,j}} / B}$ \IOs.
For splitting around a rank, deterministic selection~\cite{BlumFloydPrattRivestTarjan1973} is applied, which uses $\Oh{\lrSize{\Int{i,j}} / B}$~\IOs (see, e.g.,~\cite{BrodalWild2023,BrodalWild2024}).
In both cases the number of \IOs grows with the interval size.
Analyzing the change in the potential function upon this split, we can show that \SplitInterval uses amortized $\Oh{\frac{1}{B} \log_2 \lrSize{\Delta_i} + \logB \logB \lrSize{\Delta_i} + \frac{1}{B} x \log_2 c}$ \IOs.

This (together with the details from \wref{app:interval_structure_full_version}) concludes the proof of \wref{thm:interval_structure}.

\section{Lazy B-Trees}
\label{sec:combined_structure}

In this section, we combine the gap structure of \wref{sec:gap_structure} and the interval structure of \wref{sec:interval_structure}, to achieve an external-memory lazy search-tree.
Recall our goal, \wref{thm:main_structure}.

\thmMainStructure*

We use the \IO bounds shown in \wref[Theorems]{thm:gap_structure} and~\ref{thm:interval_structure} to bound the cost of the combined operations. These operations are performed as follows.

A $\Construct(S)$ operation is performed by creating a single gap over all elements in $S$, and in the gap create a single interval with all elements.
This can be done by a single scan of~$S$, and assuming that $S$ is represented compactly/contiguously, this uses $\Oh{\lrSize{S} / B}$ \IOs.

To perform an $\Insert(e)$ operation, \SearchGap is performed to find the gap $\Delta_i$ containing element $e$.
Next, \InsertInterval is performed to insert $e$ into the interval structure of gap $\Delta_i$.
Finally, \IncrementGap is performed on $\Delta_i$, as the size has increased.
In total this uses worst-case $\Oh[\big]{\logB \frac{N}{\lrSize{\Delta_i}} + \logB \logB \lrSize{\Delta_i}}$ \IOs.

Similarly, upon a $\Delete(\mathit{ptr})$ operation, \DeleteInterval is performed using $\mathit{ptr}$, to remove the element $e$ at the pointer location.
Next \DecrementGap is performed on the gap $\Delta_i$.
In total this uses amortized $\Oh[\big]{\logB \frac{N}{\lrSize{\Delta_i}} + \frac{1}{B} \log_2 \lrSize{\Delta_i} + \logB \logB \lrSize{\Delta_i}}$~\IOs.

A \ChangeKey operation may only change an element, s.t.\ it remains within the same gap (otherwise we need to use \Delete and \Insert). This operation therefore does not need to perform operations on the gap structure, but only on the interval structure of the relevant gap. The operation is therefore performed directly using the \ChangeInterval operation.

The \QueryElement and \QueryRank operations are performed in similar fashions.
First, \SearchGap or \RankGap is performed to find the relevant gap $\Delta_i$ containing the queried element.
Next \SplitInterval is performed on the interval structure, which yields two new gaps $\Delta'_i$ and $\Delta'_{i+1}$, which is updated into the gap structure using \SplitGap.
Note that the number of gaps is bounded by the number of elements $N$, but also by the number of queries $q$ performed, as only queries introduce new gaps.
In total, the \QueryElement and \QueryRank operations uses amortized $\Oh{\logB \min \left\{ N, q \right\} + \frac{1}{B} \log_2 \lrSize{\Delta_i} + \logB \logB \lrSize{\Delta_i} + \frac{1}{B} x \log_2 c}$ \IOs.

The space usage of the gap structure is $\Oh{G / B} = \Oh{N / B}$ blocks.
For gap $\Delta_i$, the space usage of the interval structure is $\Oh{\lrSize{\Delta_i} / B}$ blocks.
If $\lrSize{\Delta_j} = \Oh{B}$, for some $\Delta_j$, we cannot simply sum over all the substructures. By instead representing all such ``lightweight'' gaps in a single blocked-linked-list, we obtain that the space usage over the combined structure is $\Oh{N / B}$ blocks. 
The gaps are located in the combined list using the gap structure, and updates or queries to the elements of the gap may then be performed using $\Oh{1}$ \IOs, allowing the stated time bounds to hold. Similarly, an interval structure may be created for a gap, when it contains $\OhMega{B}$ elements using $\Oh{1}$ \IOs. The stated time bounds therefore still applies to the lightweight gaps.

This concludes the proof of \wref{thm:main_structure}.

\subparagraph{Priority Queue Case.}
When a lazy B-tree is used a priority queue, the queries are only performed on the element of smallest rank, and this element is deleted before a new query is performed.
If this behavior is maintained, we first note that the structure only ever has a single $1$-sided gap of size $N$, with the queried side to the left.
This allows for \Insert to be performed in worst-case $\Oh{\logB \logB N}$ \IOs, as the search time in the gap structure reduces to $\Oh{1}$.
Similarly, \Delete is performed using amortized $\Oh{\frac{1}{B} \log_2 N + \logB \logB N}$ \IOs.
To perform the \DecreaseKey operation, a \ChangeKey operation is performed. As the queried side is to the left, this must move the element towards the closest queried side, leading to a \DecreaseKey operation being performed using worst-case $\Oh{\logB \logB N}$ \IOs.
Finally, \Minimum is performed as a $\QueryRank(r)$ with $r=1$. As this must split the gap at $x = 1$ and $cx = N-1$, the operation is performed using amortized $\Oh{\frac{1}{B} \log_2 N + \logB \logB N}$ \IOs.
This concludes the proof of \wref{cor:priority_queue}.

\section{Open Problems}
\label{sec:open_problems}

The internal lazy search trees~\cite{SandlundWild20} discuss further two operations on the structure; \emph{\textsc{Split}} and \emph{\textsc{Merge}}, which allows for splitting or merging disjoint lazy search tree structures around some query element or rank. These operations are supported as efficient as queries in the internal-model.
In the external-memory case, the gap structure designed in this paper, however, does not easily allow for splitting or merging the set of gaps around an element or rank, as the gaps are stored by element in disjoint buckets, where all buckets must be updated upon executing these extra operations.
By sorting and scanning the gaps, the operations may be supported, but not as efficiently as a simple query, but instead in sorting time relative to the number of gaps.
It remains an open problem to design a gap structure, which allows for efficiently supporting \textsc{Split} and \textsc{Merge}.

In the external-memory model, bulk operations are generally faster, as scanning consecutive elements saves a factor $B$ \IOs. One such operation is \emph{\textsc{RangeQuery}}, where a range of elements may be queried and reported at once. In a B-tree, this operation is supported in $\Oh{\logB N + k/N}$ \IOs, when $k$ elements are reported.
The lazy B-tree designed in this paper allows for efficiently reporting the elements of a range in unsorted order, by querying the endpoints of the range, and then reporting all elements of the gaps between the endpoints.
Note that for the priority-queue case, this allows reporting the $k$ smallest elements in $\Oh{\frac{1}{B}k \log_2 \frac{N}{k} + \frac{1}{B} \log_2 N + \logB \logB N}$ \IOs.
If the elements must be reported in sorted order, sorting may be applied to the reported elements. However, this effectively queries all elements of the range, and should therefore be reflected into the lazy B-tree. As this introduces many gaps of small size, the \IO cost increases over simply sorting the elements of the range.
It remains an open problem on how to efficiently support sorted \textsc{RangeQuery} operations, while maintaining the properties of lazy B-trees.

In the internal-memory model, an optimal version of lazy search trees was constructed~\cite{SandlundZhang22}, which gets rid of the added $\log \log N$ term on all operations.
It remains an open problem to similarly get rid of the added $\logB \logB N$ term for the external-memory model.

Improvements on external-memory efficient search trees and priority queues use buffering of updates to move more data in a single \IO, thus improving the \IO cost of operations.
The first hurdle to overcome in order to create buffered lazy B-trees is creating buffered biased trees used for the gap structure. By buffering updates to the gaps, it must then hold that the weights are not correctly updated, which imposes problems on searching; both by element and rank.
It remains an open problem to overcome this first hurdle as a step towards buffering lazy B-trees.

\bibliographystyle{plainurl}
\bibliography{references}

\newpage
\appendix

\section{The Interval Structure}
\label{app:interval_structure_full_version}

In this appendix, we give the full description and detailed analysis of the interval data structure
outlined in \wref{sec:interval_structure}.
We follow the construction of Sandlund and Wild~\cite{SandlundWild20}, in their design of the interval structure, with some improvements for external memory.
To give a self-contained analysis of the number of \IOs used by all operations,
we reproduce the analysis from~\cite{SandlundWild20,SandlundWild20arxiv}
and spell out all cases explicitly.

\textbf{Note that we neglect the cost for keeping external pointers to elements up to date.}
When splitting or fusing nodes of the blocked-linked list, we have to update pointers to elements that have moved between blocks. Since these external pointers may reside scattered across memory, we may have to pay one \IO per pointer, which increases the worst-case cost of splits by a factor of $B$ (using potential $1 + \max \left\{ \lrSize{\Int{i, j}} - \out{i, j} , 0 \right\} \}$ for $\Phi_{i,j}$).

A gap $\Delta_i$ contains all elements in the range the gap covers.
Recall that a query on the gap for an element must be faster for elements closer to the border of a gap. 
We allow for the following operations:

\begin{description}
    \item[$\InsertInterval(e)$:]
        Inserts element $e$ into the structure.
    \item[$\DeleteInterval(\mathit{ptr})$:]
        Deletes the element $e$ at pointer $\mathit{ptr}$ from the structure.
    \item[$\ChangeInterval(\mathit{ptr}, e')$:]
        Changes the element $e$ at pointer $\mathit{ptr}$ to the element $e'$.
    \item[$\SplitInterval(e)$ or $\SplitInterval(r)$:]
        Splits the set of intervals into two sets of intervals at element $e$ or rank $r$.
\end{description}

A gap has a \emph{``sidesness''}, which denotes the number of sides the gap has had a query. 
Denote a side as a \emph{queried side}, if that rank has been queried before (cf.\ \cite{SandlundWild20}).
If there have been no queries yet, the (single) gap is a $0$-sided gap.
When a query in a $0$-sided gap occurs, two new gaps are created which are both $1$-sided. 
Note that both gaps which have a single query on either the left or right side is called $1$-sided. 
When queries have touched both sides, the gap is $2$-sided.

We obtain the following \IO costs for the above operations.
\begin{theorem}
\label{thm:apx_interval_structure}
    There exists a data structure on an ordered set, that maintains a set of intervals, supporting 
    \begin{itemize}
        \item \InsertInterval in worst-case $\Oh{\logB \logB \lrSize{\Delta_i}}$ \IOs, 
        \item \DeleteInterval in amortized $\Oh{\frac{1}{B} \log_2 \lrSize{\Delta_i} + \logB \logB \lrSize{\Delta_i}}$ \IOs, 
        \item \ChangeInterval in worst-case $\Oh{\logB \logB \lrSize{\Delta_i}}$ \IOs, if the element is moved towards the nearest queried side or amortized $\Oh{\frac{1}{B} \log_2 \lrSize{\Delta_i} + \logB \logB \lrSize{\Delta_i}}$ \IOs otherwise, and
        \item \SplitInterval in amortized $\Oh{\frac{1}{B} \log_2 \lrSize{\Delta_i} + \logB \logB \lrSize{\Delta_i} + \frac{1}{B} x \log_2 c}$ \IOs.
    \end{itemize}
    Here $\lrSize{\Delta_i}$ denotes the number of elements contained in all intervals, and $x$ and $cx$ for $c \ge 1$ are the resulting sizes of the two sets produced by a split.
    The space usage is $\Oh{\lrSize{\Delta_i} / B}$~blocks.
\end{theorem}

The structure is as follows.
Let the gap $\Delta_i$ contain multiple non-overlapping intervals $\Int{i, j}$, which contains the elements located in the gap.
The elements of the intervals are sorted between intervals, but not within an interval. Intervals therefore span a range of elements with known endpoints.
Each such interval is a blocked-linked-list containing the elements of the interval.
Additionally, we store an augmented B-tree over the intervals, allowing for efficiently locating the interval containing a given element (using $\Oh{\logB(\#\text{intervals})}$ \IOs).
The B-tree is augmented to hold in each node the total size of all intervals in the subtree of the node, which allows for efficient rank queries.

By packing all intervals into a single blocked-linked-list, and noting that there can be no more intervals than there are elements, the space usage of the intervals and the augmented B-tree over the intervals is $\Oh{\lrSize{\Delta_i} / B}$ blocks.

\subsection{Invariants and Potential}

Intuitively, there is a trade-off in maintaining intervals: having many small intervals reduces future query costs (typically dominated by the linear cost of splitting one interval), but increases the time to search for the correct interval upon insertions (and other operations).
Lazy search trees handle this trade-off as follows.
To bound the worst case insertion costs, we enforce a hard limit number on the number of intervals via a merging rule.
To amortize occasional high costs for queries, we accrue potential ($\Phi_{i,j}$ below) for any intervals that are ``too large'' relative to their proximity of a gap boundary.

Moreover, for queries we would like to have interval sizes grow exponentially towards the middle, as this allows for efficiently manipulating the intervals closest to the edge. 
This folklore rule seems too rigid to maintain upon queries, though, as many changes can be necessary.

\subsubsection{Merging Intervals}

We maintain the following invariant on the number of intervals:

\begin{description}
\item[{Invariant (\#Int):}]
	For a gap $\Delta_i$, the number of intervals $I_{i,j}$ is at most $4\log_2 \lrSize{\Delta_i}+2$. 
\end{description}

This allows for locating an interval using $\Oh{\logB \log_2 \lrSize{\Delta_i}} = \Oh{\logB \logB \lrSize{\Delta_i}}$ \IOs. This \IO cost also applies to recomputing the sizes of subtrees when updating an interval.
To bound the \IO cost of the operations, we define the following for the analysis.
Let $\out{i, j}$, the ``outside (elements) of $I_{i,j}$'', denote the number of elements between interval $\Int{i, j}$ and a \emph{queried} side of the gap; if both sides have been queries, $\out{i,j}$ is the towards the \emph{closest} side.
We denote an interval to be \emph{left} if the closest edge is to the left and \emph{right} otherwise.
Note for a $0$-sided gap, the outside is always $0$.  For a $1$-sided gap with a query on the left (resp.\ right), all intervals are left (resp.\ right).

When an insertion occurs in the gap, the interval containing the inserted element is located, and the inserted element is appended to the list of the interval. 
This increases $\lrSize{\Delta_i}$ by one, but does not create any new intervals, so (\#Int) remains satisfied.

When a deletion occurs, $\lrSize{\Delta_i}$ decreases, which may invalidate (\#Int).
We therefore need a way to reduce the number of intervals,
which we do by simply \emph{merging} two adjacent intervals.
Recall that elements within intervals are not sorted, so merging two intervals effectively just means ``forgetting'' about the pivot element currently separating them; here, we need to concatenate the list of elements.
The rule in lazy search trees is as follows, using the outside as a measure for how close to the middle an interval is:

\begin{description}
\item[Merge Rule (M):]
	If \( \lrSize{\Int{i, j}} + \lrSize{\Int{i, j+1}} < \out{i, j}\), merge $I_{i,j}$ and $I_{i,j+1}$.\\
	This rule only applies between intervals of the same type; a left interval is never merged with a right interval.
\end{description}

That is, an interval $\Int{i, j}$ is merged with its inner neighbor, if they combined contain fewer elements than outside of the intervals; they are therefore intuitively ``too small'' for where they reside in the gap.
Note that we do \emph{not} treat this rule as an invariant; we do not require it to be true at all times, 
but only use it to reduce the number of intervals when explicitly invoked in operation \MergeIntervals.

Suppose that \MergeIntervals has just merged all interval pairs that satisfy (M).
We then obtain bounds on the number of intervals as stated in the following two lemmas.

\begin{lemma}[The Merge rule on one side yields logarithmic internals]
\label{lem:merge_interval_onesided_count_bound}
    Let a set of $N$ elements partitioned in intervals that satisfy rule (M).
    Let further $I$ be one of the intervals and $k = \out{}$ be its number of outside elements.
    Then there are at most $\max\{1, 2 \log_2 \frac{N}{k} + 2\}$
    intervals inside of $I$ and on the same side as $I$.
\end{lemma}
\begin{proof}
    The proof goes by induction in $N$. When $N = 1$, there can be at most one interval. If $N < k$, there can be at most one interval, as the merge rule would otherwise be violated. In both cases, the claim is satisfied.
    
    For $N \ge k$, we may assume the induction hypothesis on all smaller values of $N$.
    If there are less than three intervals, the bound holds from the additive constant.
    Otherwise, let $a$, $b$ and $c$ be the size of the first three intervals in order.
    By the merge rule, it holds that $a + b \ge k$. By definition $\out{c} = k + a + b \ge 2 k$.
    The number of intervals is then the two intervals $a$ and $b$, and the number of intervals from $c$, which by induction is $\le 2 + 2 \log_2 \frac{N - a - b}{\out{c}} + 2 \le 2 \log_2 \frac{N}{2k} + 4 = 2 \log_2 \frac{N}{k} + 2$.
    This concludes the proof.
\end{proof}

\begin{lemma}[(M) implies (\#Int)]
\label{lem:merge_intervals_gap_count}
    Let a gap $\Delta_i$ adhere to rule (M) on intervals.
    Then there are at most $4 \log_2 \lrSize{\Delta_i} + 2$ intervals.
\end{lemma}
\begin{proof}
    The proof goes by case analysis on the sidesness of $\Delta_i$.
    If the gap is $0$-sided, there must be one interval.
    If the gap is $1$-sided, the first interval has size at least 1, and there must then be at most $2 \log \Delta_i + 2$ remaining intervals by \wref{lem:merge_interval_onesided_count_bound}.
    In both cases the bound holds.

    Otherwise, the gap must be $2$-sided, and the merge rule holds from both left and right. The outermost intervals have size at least 1. Let $L$ and $R$ denote the number of elements in left and right intervals respectively.
    By \wref{lem:merge_interval_onesided_count_bound}, the number of intervals is at most
    \[ 
    	2 + \lrParens{2 \log_2 L + 2} + \lrParens{2 \log_2 R + 2} 
    	\;=\; 
    	6 + 2 \log_2 \lrParens{ L \cdot R } 
    	\;\le\; 
    	6 + 2 \log_2 \frac{\Delta_i^2}{4} 
    	\;=\; 
    	4 \log_2 \Delta_i + 2 \; . \popQED 
    \]
\end{proof}

\subsubsection{Potential}

When a query occurs, an interval $\Int{i, j}$ is split around some element or rank. The interval can be located using the B-tree over intervals on either element or rank in $\Oh{\logB \logB \lrSize{\Delta_i}}$ \IOs.
For splitting around an element, the interval is scanned and partitioned using $\Oh{\lrSize{\Int{i,j}} / B}$ \IOs.
For splitting around a rank, deterministic selection~\cite{BlumFloydPrattRivestTarjan1973} is applied, which uses $\Oh{\lrSize{\Int{i,j}} / B}$ \IOs (see, e.g., \cite{BrodalWild2023,BrodalWild2024}).
In both cases the number of \IOs grows with the interval size.

We therefore need a potential function, which balances the time to split an interval with how close to the border it is, which is defined using the outside of an interval.
For $1$-sided gaps, the outside however only measures toward one side of the gap, and we therefore need further potential for the $1$-sided intervals. Let $N_{01}$ be the number of elements contained in total in $0$- and $1$-sided gaps.
Further, we need potential to efficiently merge intervals in a scan.

We let the potential be
\[ \Phi = \frac{1}{B} N_{01} + \sum_{i, j} \Phi_{i,j} \]
and let
\[ \Phi_{i,j} = 1 + \frac{1}{B} \max \left\{ \lrSize{\Int{i, j}} - \out{i, j} , 0 \right\} \; . \]

Using this potential function, we can analyze the amortized \IO cost of applying the merge rule on all intervals, which we denote as the \MergeIntervals operation.
This operation uses that scanning all keys of a B-tree can be done efficiently.

\begin{lemma}[Merging runtime]
\label{lem:merge_runtime}
    Let the number of intervals contained in the gap $\Delta_i$ be $k$.
    Then \MergeIntervals uses amortized $\Oh{k/B}$ \IOs.
\end{lemma}
\begin{proof}
    When \MergeIntervals is performed, first $\out{i, j}$ is computed for each $\Int{i,j}$ in $\Oh{k/B}$ \IOs, by scanning the B-tree. The B-tree may then be scanned again, to perform the merges on the intervals, and finally a new B-tree over the resulting intervals can be built. As the number of intervals can only decrease, the B-tree can be build using $\Oh{k/B}$ \IOs.

    Let $\Int{i,j}$ and $\Int{i,j+1}$ be a pair of intervals which are merged during the operation.
    It must hold that $\lrSize{\Int{i,j}} + \lrSize{\Int{i,j+1}} < \out{i,j}$.
    The merge is performed by following the pointers to the intervals and concatenating the lists, using $\Oh{1}$ \IO.
    Let $I_{i, j'}$ be the resulting interval of the merge. The potential then decreases by $\Phi_{i, j}$ and $\Phi_{i, j+1}$, and increases by $\Phi_{i, j'}$.
    From the merge rule and the definition of the outside, it follows that
    \begin{alignat*}{3}
        &\Phi_{i, j} &&= 1 + \frac{1}{B} \max \left\{ \lrSize{\Int{i,j}} - \out{i,j} , 0 \right\} &&= 1 \\
        &\Phi_{i, j+1} &&= 1 + \frac{1}{B} \max \left\{ \lrSize{\Int{i,j+1}} - \lrParens{\out{i,j} + \lrSize{\Int{i,j}}} , 0 \right\} &&= 1 \\
        &\Phi_{i, j'} &&= 1 + \frac{1}{B} \max \left\{ \lrParens{ \lrSize{\Int{i,j}} + \lrSize{\Int{i,j+1}}} - \out{i,j} , 0 \right\} &&= 1 \; .
    \end{alignat*}
    The difference in potential is therefore $\Delta \Phi = -1$, which covers the \IOs performed by the merge, which concludes the proof.
\end{proof}

\subsection{Updating Elements of the Intervals}

Upon an \InsertInterval operation, a new element $e$ is added to the gap $\Delta_i$. The interval $\Int{i, j}$ containing $e$ can be located using the B-tree, and the size of the interval incremented by $1$ using $\Oh{\logB \logB \lrSize{\Delta_i}}$ \IOs.
As previously discussed, the intervals need not be merged, as the invariant on the interval count is satisfied.
In the potential function $\Phi$, the outside of all inner intervals grows by $1$, which does not increase the value of $\Phi$. The size of $\Int{i, j}$ grows by $1$, and $N_{01}$ grows by at most $1$, which in total increases $\Phi$ by at most $\frac{2}{B}$.
\InsertInterval can therefore be performed in both amortized and worst-case $\Oh{\logB \logB \lrSize{\Delta_i}}$ \IOs.

When an \DeleteInterval operation occurs, a pointed to element is removed from some interval $\Int{i,j}$. The size is then decremented by $1$, and the B-tree must be updated to reflect so. If this decreases the size to $0$, the interval is removed, which in total uses $\Oh{\logB \logB \lrSize{\Delta_i}}$ \IOs.
In contrast to the \InsertInterval operation, upon removing an element, the outside of all inner intervals decreases, which makes the potential function $\Phi$ grow by at most $\frac{1}{B}$ for each such interval, of which there are $\Oh{\log_2 \lrSize{\Delta_i}}$.
In addition, the number of elements decreases without the number of intervals decreasing, which may leading to breaking the invariant on the bound on the number of intervals.
Note that merging a single interval is enough to ensure that the number of intervals stays bounded, and even so, only after linearly many deletes, the merge of a single interval needs to occur.
However, as the potential increases by $\Oh{\frac{1}{B} \log_2 \lrSize{\Delta_i}}$, and the \MergeIntervals uses the same amount of amortized \IOs, it is efficient enough to run \MergeIntervals after every delete.
In total, \DeleteInterval uses amortized $\Oh{\frac{1}{B} \log_2 \lrSize{\Delta_i} + \logB \logB \lrSize{\Delta_i}}$ \IOs.

When $\ChangeInterval(\mathit{ptr}, e')$ in performed, the element $e$ at pointer $\mathit{ptr}$ is changed to element $e'$. This new element must then be placed in the correct interval, which can be done using the B-tree over intervals using $\Oh{\logB \logB \lrSize{\Delta_i}}$ \IOs.
This then changes the size of at most two intervals, with one increasing and the other decreasing, which increases the potential by at most $\frac{1}{B}$.
This may also change the outside size of intervals. If the element is moved towards the closest queried side, then the outside may increase, which does not increase the potential.
Otherwise, the outside may grow by one, for each interval the element is moved past, which increases the potential by at most $\Oh{\frac{1}{B} \log_2 \lrSize{\Delta_i}}$.
In total, \ChangeInterval uses worst-case $\Oh{\logB \logB \lrSize{\Delta_i}}$ \IOs, if the element is moved towards the closest queried side, and otherwise amortized $\Oh{\frac{1}{B} \log_2 \lrSize{\Delta_i} + \logB \logB \lrSize{\Delta_i}}$ \IOs.

\subsection{Splitting the Intervals}

When a query occurs in gap $\Delta_i$, the gap is split at the query element into two new gaps $\Delta'_i$ and $\Delta'_{i+1}$.
Let the resulting gap sizes be $x$ and $cx$ for some $c \ge 1$ (see \wref{fig:query_split_x_cx}).
We assume here and throughout that the left resulting gap has size $x$; the other case is symmetric.

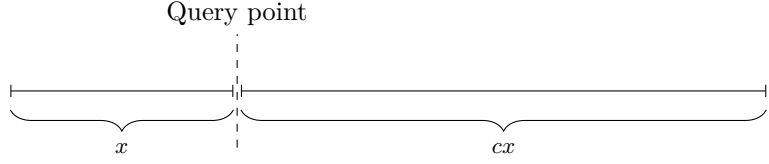
\begin{figure}
    \centering
    
    \begin{tikzpicture}
        \draw[|-|] (0, 0) -- (2.95, 0);
        \draw[|-|] (3.05, 0) -- (10, 0);
        \draw[dashed] (3,-0.75) -- (3,0.75) node[above] {Query point};
    
        \draw[decorate,decoration={brace,amplitude=8pt,mirror}] (0, -0.25) -- (2.95, -0.25) node[below, midway, yshift=-9pt] {\footnotesize $x$};
        \draw[decorate,decoration={brace,amplitude=8pt,mirror}] (3.05, -0.25) -- (10, -0.25) node[below, midway, yshift=-9pt] {\footnotesize $c x$};
    \end{tikzpicture}
    
    \caption{The query splits the gap in two, where some side contains $x$ elements and the other contains $cx$ elements, for some $c \ge 1$.}
    \label{fig:query_split_x_cx}
\end{figure}

Upon a \SplitInterval operation, the interval $\Int{i,j}$ containing the query element must be split, and the remaining intervals distributed into the two new gaps.
First, \MergeIntervals is applied on gap $\Delta_i$, to ensure structure on the intervals.
Then interval $\Int{i,j}$ is split around the query element into a left and right side of elements with respectively smaller and larger elements.
Let $\ell$ be the number of elements on the left side.
Using deterministic selection, the left side is partitioned into three intervals of sizes $\ell/2$, $\ell/4$ and $\ell/4$, such that the sizes are doubling from the query point and outwards. The right side is similarly split, resulting in six new intervals covering the elements of $\Int{i, j}$.
See \wref{fig:splitting_interval_on_query} for an illustration of this split.
From these new intervals, along with the remaining intervals of gap $\Delta_i$, two new gaps $\Delta'_i$ and $\Delta'_{i+1}$ are created, containing the intervals to the left and right side of the query point respectively.
To finally ensure invariant (\#Int) on the two new gaps, we lastly apply \MergeIntervals on them both, which concludes the steps of the \SplitInterval operation.

\begin{figure}
    \centering
    \begin{tikzpicture}
        \node[right] at (-3, 4) {Interval before query:};
        \node[right] at (-3, 1) {Intervals after query:};
    
        \draw[|-|] (0,3) -- (10,3) node[above, midway, yshift=3pt] {$\Int{i, j}$};
        \draw[dashed] (4,3.5) -- (4,-0.75) node[midway, fill=white, yshift=10pt] {Query point};

        \draw[|-|] (0,0) -- (1.95,0) node[below, midway, yshift=-3pt] {\footnotesize $\ell/2$};
        \draw[|-|] (2.05,0) -- (2.95,0) node[below, midway, yshift=-3pt] {\footnotesize $\ell/4$};
        \draw[|-|] (3.05,0) -- (3.95,0) node[below, midway, yshift=-3pt] {\footnotesize $\ell/4$};
        \draw[|-|] (4.05,0) -- (5.45,0) node[below, midway, yshift=-3pt] {\footnotesize $r/4$};
        \draw[|-|] (5.55,0) -- (6.95,0) node[below, midway, yshift=-3pt] {\footnotesize $r/4$};
        \draw[|-|] (7.05,0) -- (10,0) node[below, midway, yshift=-3pt] {\footnotesize $r/2$};

        \draw[decorate,decoration={brace,amplitude=8pt}] (0, 0.25) -- (3.95, 0.25) node[above, midway, yshift=6pt] {\footnotesize Left side, size $\ell$};
        \draw[decorate,decoration={brace,amplitude=8pt}] (4.05, 0.25) -- (10, 0.25) node[above, midway, yshift=6pt] {\footnotesize Right side, size $r$};
    \end{tikzpicture}
    
    \caption{A query in interval $\Int{i,j}$ splits the interval into six new intervals.}
    \label{fig:splitting_interval_on_query}
\end{figure}
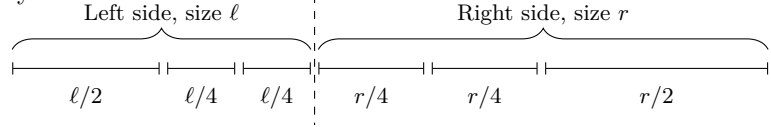

We next analyze the amortized \IO cost of \SplitInterval.
The interval $\Int{i,j}$ can be located using the B-tree in $\Oh{\logB \logB \lrSize{\Delta_i}}$~\IOs.
Splitting the interval is done using a constant number of scans and deterministic selections on the interval, which uses $\Oh{\frac{1}{B} \lrSize{\Int{i, j}}}$~\IOs in total.
When performing \MergeIntervals on $\Delta_i$, $\Delta'_i$, and $\Delta'_{i+1}$, then by \wref{lem:merge_runtime}, the invariant on the number of intervals (\#Int), and that splitting interval $\Int{i,j}$ creates $\Oh{1}$ new intervals, it holds that these merges uses amortized $\Oh{\frac{1}{B} \log_2 \lrSize{\Delta_i}}$~\IOs in total.
Note that these costs is independent on where $\Int{i, j}$ resides in the structure, and the sidesness of gap $\Delta_i$.
The change in potential however depends heavily on these properties. We shall therefore use case analysis for the change in potential.
These cases are on if the split interval is on the left or right, and on the sidesness of the gap $\Delta_i$.
We only analyze the difference in potential under the splitting of interval $\Int{i,j}$, as the \MergeIntervals operation already is amortized.

\subsubsection{\boldmath Gap $\Delta_i$ is $2$-sided.}

The split must result in two new $2$-sided gaps. Therefore, $N_{01}$ is unchanged, and can be disregarded for the potential difference.
As the gap is $2$-sided, $\out{i,j}$ is the smaller distance to the two sides of the gap.
It holds that $\Int{i,j}$ overlaps the query point, and therefore the distance to the left side (in $\Delta_i$) is at most $x$, and the distance to the right side (in $\Delta_i$) is at most $cx$. Therefore, $\out{i,j} \le x$.

There are three types of intervals: the intervals resulting from splitting $\Int{i,j}$, the intervals to the left of $\Int{i,j}$ in $\Delta_i$, and the intervals to the right of $\Int{i,j}$ in $\Delta_i$.
Let us first consider the potential on the intervals resulting from splitting $\Int{i,j}$. 
We denote this change in potential as $\Delta \Phi_I$.
The potential on $\Int{i,j}$ is $\Phi_{i,j}$, which is removed.
Then six new intervals are introduced.
Consider the intervals created on the left side of the query point, and let the total number of elements in these three intervals be $\ell$.
The rightmost one of these intervals is next to a query point in the new gap, and thus its outside is $0$. The potential on this interval is therefore $\frac{1}{B} \ell/4$.
The middle interval has the right interval as its outside, and as they are both of size $\ell/4$, there is no potential on this interval.
The leftmost interval may have either side as its outside; in the worst-case, it has a potential of $\frac{1}{B} \ell/2$.
This holds symmetrically for the new intervals on the right. As the size of the left and right in total is $\lrSize{\Int{i, j}}$, the total worst-case potential on the new intervals is $\frac{3}{4} \frac{1}{B} \lrSize{\Int{i, j}}$.
In total, the change in potential on the intervals created from $\Int{i, j}$ is
\begin{align*}
    \Delta \Phi_I \;\le\; & \; 5 && \text{The number of intervals created vs removed} \\
    & {}-  \frac{1}{B} \max \left\{ \lrSize{I_{i,j}} - o\lrParens{I_{i,j}} , 0 \right\} && \text{Removing $I_{i, j}$} \\
    & {}+ \frac{3}{4} \frac{1}{B} \lrSize{I_{i,j}} && \text{Introducing new intervals}
\end{align*}

To simplify this, we distinguish cases on the size of $\Int{i, j}$, and compare it to the \IO cost of splitting the interval.
If $\lrSize{I_{i,j}} \le x$, then $\Delta \Phi \le 5 + \frac{3}{4} \frac{1}{B} x = \Oh{\frac{1}{B} x}$. Otherwise, $\Delta \Phi \le 5 + \frac{1}{B} x - \frac{1}{4} \frac{1}{B} \lrSize{I_{i,j}}$.
The amortized cost of the split, including only the potential on the intervals of $\Int{i,j}$, is therefore $\Oh{\frac{1}{B} x}$ \IOs.
The total amortized cost must also cover the remaining intervals.

For the intervals to the left of $\Int{i,j}$, let the change in potential be denoted as $\Delta \Phi_L$.
In total there are at most $x$ elements located on the left side. As there is no new intervals introduced, the potential function on the left intervals is therefore bounded by the number of elements, and thus $\Delta \Phi_L \le \frac{1}{B} x$.
For the change in potential for the intervals on the right of $\Int{i,j}$, denoted $\Delta \Phi_R$, we must further distinguish whether $\Int{i, j}$ is a left or right interval.

\subparagraph{Case 1: $\Int{i,j}$ is a left interval.}

We partition the intervals on the right into three categories:
the intervals that were a left interval in $\Delta_i$ and remain a left interval in $\Delta'_{i+1}$ ($LL$), the intervals that were right intervals in $\Delta_i$ and become left intervals in $\Delta'_{i+1}$ ($RL$), and the intervals that were right intervals in $\Delta_i$ and remain right intervals in $\Delta'_{i+1}$ ($RR$).
Note that there cannot be any intervals that were left and became right, as the number of elements to the left decreases; if this was the side with the fewer elements before, it must remain so.
These categories are illustrated on \wref{fig:twosided_rightside_interval_categories}.

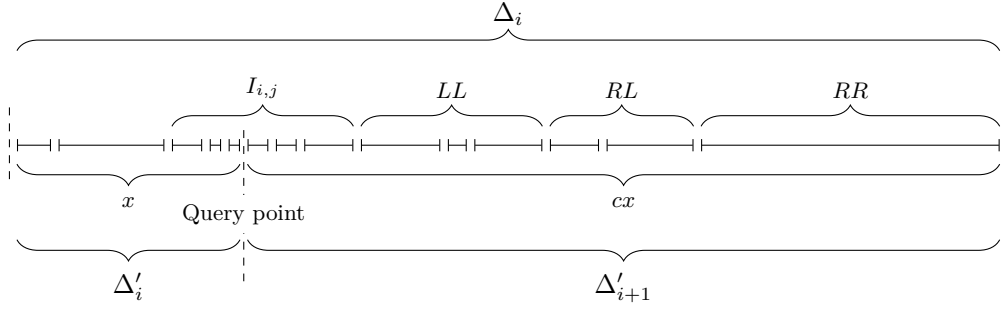
\begin{figure}
    \centering
    
    \begin{tikzpicture}
        \draw[|-|] (0, 0) -- (0.45, 0);
        \draw[|-|] (0.55, 0) -- (1.95, 0);

        \foreach \l/\r in {2/2.5, 2.5/2.75, 2.75/3, 3/3.375, 3.375/3.75, 3.75/4.5}
            \draw[|-|] (\l+0.05, 0) -- (\r-0.05, 0);

        \draw[|-|] (4.55, 0) -- (5.6, 0);
        \draw[|-|] (5.7, 0) -- (5.95, 0);
        \draw[|-|] (6.05, 0) -- (6.95, 0);

        \draw[|-|] (7.05, 0) -- (7.7, 0);
        \draw[|-|] (7.8, 0) -- (8.95, 0);

        \draw[|-|] (9.05, 0) -- (13, 0);
    
        \draw[decorate,decoration={brace,amplitude=8pt}] (2.05, 0.25) -- (4.45, 0.25) node[above, midway, yshift=7pt] {\footnotesize $\Int{i,j}$};
        \draw[decorate,decoration={brace,amplitude=8pt}] (4.55, 0.25) -- (6.95, 0.25) node[above, midway, yshift=7pt] {\footnotesize $LL$};
        \draw[decorate,decoration={brace,amplitude=8pt}] (7.05, 0.25) -- (8.95, 0.25) node[above, midway, yshift=7pt] {\footnotesize $RL$};
        \draw[decorate,decoration={brace,amplitude=8pt}] (9.05, 0.25) -- (13, 0.25) node[above, midway, yshift=7pt] {\footnotesize $RR$};
        
        \draw[dashed] (3, 0.2) -- (3, -1.8) node[midway, fill=white, yshift=-3.2pt] {\footnotesize Query point};
        \draw[dashed] (-0.1, 0.5) -- (-0.1, -0.5);
        \draw[dashed] (13.1, 0.5) -- (13.1, -0.5);
    
        \draw[decorate,decoration={brace,amplitude=8pt,mirror}] (0, -0.25) -- (2.95, -0.25) node[below, midway, yshift=-9pt] {\footnotesize $x$};
        \draw[decorate,decoration={brace,amplitude=8pt,mirror}] (3.05, -0.25) -- (13, -0.25) node[below, midway, yshift=-9pt] {\footnotesize $c x$};

        \draw[decorate,decoration={brace,amplitude=8pt}] (0, 1.25) -- (13, 1.25) node[above, midway, yshift=6pt] {$\Delta_{i}$};
        
        \draw[decorate,decoration={brace,amplitude=8pt,mirror}] (0, -1.25) -- (2.95, -1.25) node[below, midway, yshift=-9pt] {$\Delta'_i$};
        \draw[decorate,decoration={brace,amplitude=8pt,mirror}] (3.05, -1.25) -- (13, -1.25) node[below, midway, yshift=-9pt] {$\Delta'_{i+1}$};
    \end{tikzpicture}
    
    \caption{The categories of intervals upon splitting a $2$-sided gap, with $\Int{i,j}$ being a left interval.}
    \label{fig:twosided_rightside_interval_categories}
\end{figure}

First note that the number of intervals and their sizes do not change from $\Delta_i$ to $\Delta'_{i+1}$. The only change is in the size of the outsides.
For the intervals in $RR$, their outside does not change, and therefore their potential does not change.

Note that \MergeIntervals was performed before the split, and therefore rule (M) is satisfied on the intervals in $\Delta_i$. Also note that splitting another interval does not alter the outside values of other intervals.
Therefore, for the intervals in $LL$, $x$ elements are removed from their outside in $\Delta'_{i+1}$. We must then bound the number of intervals in $LL$, to bound the change in the potential.
As the outside of the leftmost interval in $LL$ has an outside of at least $x$ in $\Delta_i$, and there are at most $cx$ elements in $LL$, as they are all contained in $\Delta'_{i+1}$, then by \wref{lem:merge_interval_onesided_count_bound}, the number of intervals in $LL$ is at most $2 \log_2 \frac{cx}{x} + 2 = \Oh{\log_2 c}$.
The potential increase of $LL$ is therefore bounded by $\Oh{\frac{1}{B} x \log_2 c}$.

For the intervals in $RL$, as they were a right interval before in $\Delta_i$, but become a left interval in $\Delta'_{i+1}$ due to $x$ elements being removed from the left, their outside can decrease by at most $x$, when moved to $\Delta'_{i+1}$.
Consider all but the rightmost interval of $RL$. As these are left intervals of $\Delta'_{i+1}$ without the last left interval, it must hold that they are contained entirely in the left half of the $cx$ elements of $\Delta'_{i+1}$. As they are right intervals in $\Delta_i$, then their old outside is to the right in $\Delta_i$, and therefore the outside of the last considered interval is in $\Delta_i$ at least the right half of the $cx$ elements.
By \wref{lem:merge_interval_onesided_count_bound}, the number of intervals in $RL$, including the rightmost interval, is bounded by $2 \log_2 \frac{cx/2}{cx/2} + 2 + 1 = 3$.
The potential increase of $RL$ is therefore bounded by $\frac{3}{B} x$.

In total, $\Delta \Phi_R = \Oh{\frac{1}{B} x \log_2 c}$.

\subparagraph{Case 2: $\Int{i,j}$ is a right interval.}

As there are more elements to the left of $\Int{i,j}$ than to the right in $\Delta_i$, there are at most $x$ elements remaining to the right of $\Int{i,j}$.
As there are no new intervals introduced, the potential function on the right intervals is therefore bounded by the number of elements, and thus $\Delta \Phi_R \le \frac{1}{B} x$.

\subsubsection{\boldmath Gap $\Delta_i$ is $0$-sided.}

As there have been no queries to the left or right of the gap, then there can only be a single interval in $\Delta_i$, which is $\Int{i, j}$. The bottom of \wref{fig:splitting_interval_on_query} therefore illustrates this case.
Upon splitting this interval, two $1$-sided gaps are created. All elements are moved into these gaps from $\Delta_i$, and therefore $N_{01}$ is unchanged, and can be disregarded for $\Delta\Phi$.
For both gaps, the queried side is at the query point.
Therefore, the outside of each interval in $\Delta'_i$ and $\Delta'_{i+1}$ is towards the middle.
Following the analysis of $\Delta \Phi_I$ in the $2$-sided case, it holds that the new potential is $\frac{1}{4} \frac{1}{B} \lrSize{\Int{i,j}}$, as the outside of the outermost intervals is towards the middle, and therefore $\Delta \Phi = 5 -\frac{3}{4} \frac{1}{B} \lrSize{\Int{i,j}}$.

\subsubsection*{\boldmath Gap $\Delta_i$ is $1$-sided.}

As before, we assume for the current query, that the left resulting new gap has size $x$.
If in $\Delta_i$, the queried side was the left side, then gap $\Delta'_i$ becomes a $2$-sided gap, and $\Delta'_{i+1}$ remains a $1$-sided gap, with the queried side on the left. This case is illustrated on \wref{fig:onesided_left_interval_categories}.
The outside of $\Int{i,j}$ in $\Delta_i$ remains bounded by $x$, and the analysis on the change in potential of $\Int{i,j}$ remains the same as in the $2$-sided gap analysis, which bounds the potential increase by $\Oh{\frac{1}{B} x}$. For the intervals on the left, the potential may be at most $\frac{1}{B} x$, as in the $2$-sided case. However, $N_{01}$ decreases by $x$, which bounds the increase in potential on the left to at most $0$.
For the intervals on the right, it must hold that they were all left intervals in $\Delta_i$ and that they remain left intervals in $\Delta'_{i+1}$.
Following the argument for $LL$ from the $2$-sided gap case above, the increase in potential is bounded by $\Oh{\frac{1}{B} x \log_2 c}$.

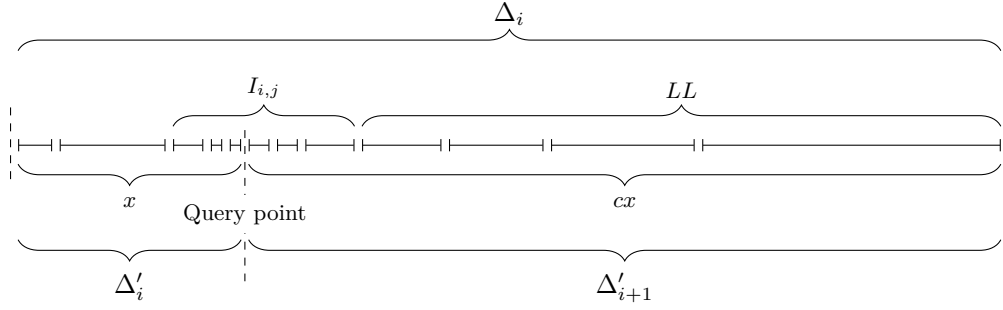
\begin{figure}
    \centering

    \begin{tikzpicture}
        \draw[|-|] (0, 0) -- (0.45, 0);
        \draw[|-|] (0.55, 0) -- (1.95, 0);

        \foreach \l/\r in {2/2.5, 2.5/2.75, 2.75/3, 3/3.375, 3.375/3.75, 3.75/4.5}
            \draw[|-|] (\l+0.05, 0) -- (\r-0.05, 0);

        \draw[|-|] (4.55, 0) -- (5.6, 0);
        \draw[|-|] (5.7, 0) -- (6.95, 0);
        \draw[|-|] (7.05, 0) -- (8.95, 0);
        \draw[|-|] (9.05, 0) -- (13, 0);
    
        \draw[decorate,decoration={brace,amplitude=8pt}] (2.05, 0.25) -- (4.45, 0.25) node[above, midway, yshift=7pt] {\footnotesize $\Int{i,j}$};
        \draw[decorate,decoration={brace,amplitude=8pt}] (4.55, 0.25) -- (13, 0.25) node[above, midway, yshift=7pt] {\footnotesize $LL$};
        
        \draw[dashed] (3, 0.2) -- (3, -1.8) node[midway, fill=white, yshift=-3.2pt] {\footnotesize Query point};
        \draw[dashed] (-0.1, 0.5) -- (-0.1, -0.5);
    
        \draw[decorate,decoration={brace,amplitude=8pt,mirror}] (0, -0.25) -- (2.95, -0.25) node[below, midway, yshift=-9pt] {\footnotesize $x$};
        \draw[decorate,decoration={brace,amplitude=8pt,mirror}] (3.05, -0.25) -- (13, -0.25) node[below, midway, yshift=-9pt] {\footnotesize $c x$};

        \draw[decorate,decoration={brace,amplitude=8pt}] (0, 1.25) -- (13, 1.25) node[above, midway, yshift=6pt] {$\Delta_{i}$};
        
        \draw[decorate,decoration={brace,amplitude=8pt,mirror}] (0, -1.25) -- (2.95, -1.25) node[below, midway, yshift=-9pt] {$\Delta'_i$};
        \draw[decorate,decoration={brace,amplitude=8pt,mirror}] (3.05, -1.25) -- (13, -1.25) node[below, midway, yshift=-9pt] {$\Delta'_{i+1}$};
    \end{tikzpicture}
    
    \caption{The categories of intervals upon splitting a $1$-sided gap, with left queried side.}
    \label{fig:onesided_left_interval_categories}
\end{figure}

\begin{figure}
    \centering

    \begin{tikzpicture}
        \draw[|-|] (0, 0) -- (1.95, 0);

        \foreach \l/\r in {2/2.5, 2.5/2.75, 2.75/3, 3/3.375, 3.375/3.75, 3.75/4.5}
            \draw[|-|] (\l+0.05, 0) -- (\r-0.05, 0);

        \draw[|-|] (4.55, 0) -- (5.95, 0);
        \draw[|-|] (6.05, 0) -- (7.7, 0);

        \draw[|-|] (7.8, 0) -- (8.95, 0);
        \draw[|-|] (9.05, 0) -- (11.95, 0);
        \draw[|-|] (12.05, 0) -- (13, 0);
    
        \draw[decorate,decoration={brace,amplitude=8pt}] (2.05, 0.25) -- (4.45, 0.25) node[above, midway, yshift=7pt] {\footnotesize $\Int{i,j}$};
        \draw[decorate,decoration={brace,amplitude=8pt}] (4.55, 0.25) -- (7.7, 0.25) node[above, midway, yshift=7pt] {\footnotesize $RL$};
        \draw[decorate,decoration={brace,amplitude=8pt}] (7.8, 0.25) -- (13, 0.25) node[above, midway, yshift=7pt] {\footnotesize $RR$};
        
        \draw[dashed] (3, 0.2) -- (3, -1.8) node[midway, fill=white, yshift=-3.2pt] {\footnotesize Query point};
        \draw[dashed] (13.1, 0.5) -- (13.1, -0.5);
    
        \draw[decorate,decoration={brace,amplitude=8pt,mirror}] (0, -0.25) -- (2.95, -0.25) node[below, midway, yshift=-9pt] {\footnotesize $x$};
        \draw[decorate,decoration={brace,amplitude=8pt,mirror}] (3.05, -0.25) -- (13, -0.25) node[below, midway, yshift=-9pt] {\footnotesize $c x$};

        \draw[decorate,decoration={brace,amplitude=8pt}] (0, 1.25) -- (13, 1.25) node[above, midway, yshift=6pt] {$\Delta_{i}$};
        
        \draw[decorate,decoration={brace,amplitude=8pt,mirror}] (0, -1.25) -- (2.95, -1.25) node[below, midway, yshift=-9pt] {$\Delta'_i$};
        \draw[decorate,decoration={brace,amplitude=8pt,mirror}] (3.05, -1.25) -- (13, -1.25) node[below, midway, yshift=-9pt] {$\Delta'_{i+1}$};
    \end{tikzpicture}
    
    \caption{The categories of intervals upon splitting a $1$-sided gap, with right queried side.}
    \label{fig:onesided_right_interval_categories}
\end{figure}
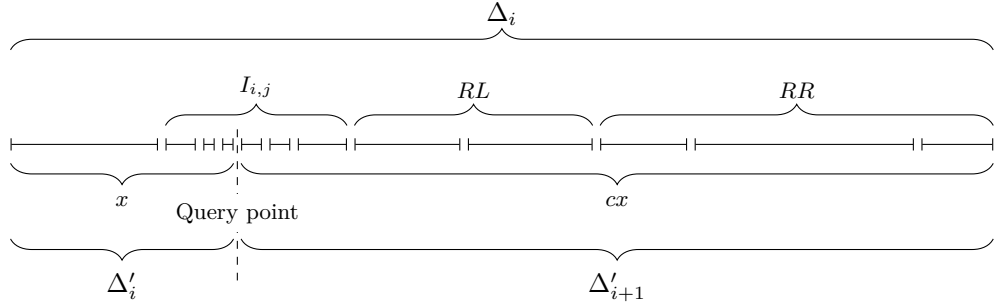

The case when the queried side in $\Delta_i$ is to the right is illustrated on \wref{fig:onesided_right_interval_categories}.
In this case, the outside of $\Int{i, j}$ is no longer bounded by $x$, but by $cx$.
However, all elements on the right leave the $1$-sided gap $\Delta_i$, and enter the $2$-sided gap $\Delta'_{i+1}$; therefore $N_{01}$ decreases by $cx$, and so the potential decreases by $\frac{1}{B} cx$.
The intervals on the right must have been right intervals in $\Delta_i$, and may be both left or right intervals in $\Delta'_{i+1}$. Using the same analysis as in the $2$-sided case for the $RL$ and $RR$ intervals, it holds that the potential on the right interval increase by at most $\frac{3}{B} x$.
The intervals to the left spans at most $x$ elements, and the potential on these are at most $\frac{1}{B} x$, as in the $2$-sided case.
In total 
\begin{align*}
    \Delta \Phi \;\le\; & \; 5 && \text{The number of intervals created vs removed} \\
    &{} -  \frac{1}{B} \max \left\{ \lrSize{I_{i,j}} - o\lrParens{I_{i,j}} , 0 \right\} && \text{Removing $I_{i, j}$} \\
    &{} + \frac{3}{4} \frac{1}{B} \lrSize{I_{i,j}} && \text{Introducing new intervals} \\
    &{} - \frac{1}{B} cx && \text{Elements removed from $1$-sided gap} \\
    &{} + \frac{3}{B} x + \frac{1}{B} x && \text{Potential change on the remaining intervals}
\end{align*}

Recall that $\out{i,j} \le cx$. We then case analysis on the size of $\Int{i,j}$.
If $\lrSize{\Int{i,j}} \le cx$, then $\Delta \Phi \le 5 - \frac{1}{4} \frac{1}{B} cx + \frac{4}{B} x$. As splitting the interval uses $\Oh{\frac{1}{B} \lrSize{\Int{i,j}}}$ \IOs, the amortized cost of this case is $\Oh{\frac{1}{B} x}$ \IOs.
Otherwise, if $\lrSize{\Int{i,j}} > cx$, then $\Delta \Phi \le 5 - \frac{1}{4} \frac{1}{B} \lrSize{\Int{i,j}} + \frac{4}{B} x$, and the amortized cost of this case is therefore also $\Oh{\frac{1}{B} x}$ \IOs.

\subsubsection{Total \IO Cost}

In total, the increase in potential on the intervals, and therefore also the amortized \IO cost of the split, is bounded by $\Oh{\frac{1}{B} x \log_2 c}$.
The \SplitInterval operation, which includes searching for the query point and the \MergeIntervals operations, therefore uses amortized $\Oh{\frac{1}{B} \log_2 \lrSize{\Delta_i} + \logB \logB \lrSize{\Delta_i} + \frac{1}{B} x \log_2 c}$ \IOs in total.

This concludes the proof of \wref{thm:apx_interval_structure}.

\section{Supporting Lemmas}
\label{app:support_lemma}

\begin{lemma}[Bounding the Sum of Double Exponential Function]
\label{lem:theta_double_exp_sum}
    Let $a, b \ge 2$. Then
    \[ \sum_{i=0}^{n} a^{b^i} = \OhTheta{a^{b^n}} \; . \]
\end{lemma}

\begin{proof}
    The proof proceeds by showing the lower and upper bound separately. For the lower bound, as all terms of the sum is positive, then the sum is bounded below by the last term of the sum.
    To show the upper bound, it is shown by induction in $n$ that
    \[ \sum_{i=0}^{n} a^{b^i} \le 2 a^{b^n} \; . \]
    For $n = 0$ the inequality holds, as there is only one term which is equal to $a$.
    \[ \sum_{i=0}^{n} a^{b^i} = \sum_{i=0}^{0} a^{b^i} = a^{b^0} = a \le 2 a = 2 a^{b^0} = 2 a^{b^n} \; .\]
    For the induction step, let $n = n' + 1$, and let the inequality hold by induction for $n'$. Then
    \[ \sum_{i=0}^{n} a^{b^i} = \left( \sum_{i=0}^{n'} a^{b^i} \right) + a^{b^{n' + 1}} \le 2 a^{b^{n'}} + \left(a^{b^{n'}} \right)^b \]
    From here, the goal bound is twice that of the last term in the sum. It can therefore be shown by proving that the second term bounds the first. For simplicity, let $k = a^{b^{n'}}$. Then it holds that
    \[ 2 k \le k^b  \Longleftrightarrow  \log_k ( 2 k ) \le b  \Longleftrightarrow  \log_k (2) + 1 \le b \; . \]
    As $a, b \ge 2$ and $n' \ge 0$, then $\log_k (2) \le 1$, and the above inequality holds. It therefore holds that
    \[ 2 a^{b^{n'}} + \left(a^{b^{n'}} \right)^b \le 2 \left(a^{b^{n'}} \right)^b = 2 a^{b^{n}} \; , \]
    Concluding the proof.
\end{proof}

\end{document}